\documentclass[reprint,superscriptaddress,nofootinbib,amsmath,amssymb]{article}
\usepackage{color}
\usepackage{amsfonts}
\usepackage{graphicx} 
\usepackage[utf8]{inputenc}
\usepackage{amsmath}
\usepackage{hyperref}
\usepackage{verbatim}
\usepackage{amsthm}
\usepackage{amssymb}
\usepackage[left=1in, right=1in, top=1in, bottom=1.5in]{geometry}
\usepackage{setspace}
\usepackage[T1]{fontenc}
\usepackage{graphicx}
\usepackage{wrapfig}

\newtheorem{theorem}{Theorem}[section]

\newtheorem{lemma}[theorem]{Lemma}

\begin{document}
	
\title{A Two Species Thomas-Fermi Model for Stellar Ground States}
	
\author{Parker Hund$^1$}

\maketitle
\begin{abstract}
\noindent 
We present an energy functional for a Thomas-Fermi type two-fluid model of a self-gravitating non-rotating charged body, with a non-relativistic kinetic energy. We prove that, under certain conditions on the total number of positively charged and negatively charged particles, a minimizer exists and both fluids have compact support. We prove the same result for special relativistic kinetic energy, assuming further conditions on the total number of particles. In the non-relativistic kinetic energy case, we further prove the uniqueness of the minimizer, as well as present results relating the general shape of the minimizer to the total number of particles.
\end{abstract}
$\phantom{xi}$\hfill

\vspace{-10pt}
\section{Introduction}\vspace{-5pt}
The Thomas-Fermi model as an approximate model was created for the multi-particle Schr\"odinger equation \cite{Thomas}, \cite{Fermi}. In its molecular formulation, it is concerned with minimizing the energy functional 
\begin{equation}\label{TFequation}
   F(\rho)= C\int_{\mathbb{R}^3}\rho^{5/3}(x)d^3x+\int_{\mathbb{R}^3}\rho(x)V(x)d^3x+\frac{1}{2}q^2\int_{\mathbb{R}^3}\int_{\mathbb{R}^3}\frac{\rho(x)\rho(y)}{|x-y|}d^3xd^3y
\end{equation}
under certain constraints on $\rho$.  Here $\rho$ represents the electrons, treated as fluid, which are in a potential field $V(x)$ created by fixed atomic nuclei. The first term in the functional represents the kinetic energy of the electrons while the last term represents the energy associated with the electrical repulsion between electrons. A very detailed discussion of the exact formulation of this model as well as a mathematically rigorous discussion of its minimizers can be found in \cite{LS}, \cite{Spruch}. 

How accurate this approximation to the energy of the ground state is was unclear at the time, but it is now known that the Thomas-Fermi energy is the leading order term, in the number of particles, of the ground state energy of an atom \cite{Fefferman}. As an approximate model it does have important limitations, however. Most well known is the result of Teller that the theory does not predict stable molecules. That is, one can lower the energy of a molecule by breaking it apart into constituent atoms and moving them far apart from each other \cite{Teller}.  This property  was instrumental in Lieb and Thirring's proof of stability of non-relativistic matter by allowing for a reduction of ``stability of the second kind'' to ``stability of the first kind'' \cite{LiebThirring}; that is, one reduces the problems to atoms instead of molecules. A discussion of this idea and how it appears in different proofs for the stability of matter can be found in \cite{LiebSeiringer}.  

In this paper, we apply the Thomas-Fermi approximation to study the structure of stellar objects instead of microscopic matter. This is somewhat similar to the recent applications of Thomas-Fermi theory found in \cite{RRRX1} and \cite{RRRX2}, although those models are not presented in a mathematically rigorous way (nor is their aim to be) and exhibit substantial differences when compared to ours. It is more accurate to say that we follow in the footsteps of the effective single fluid model for white dwarfs presented by Chandrasekhar \cite{Chandra} and the results of Lieb and Yau in \cite{LY}.

We essentially take Chandrasekhar's model and forget the local neutrality assumption. That is, he assumed that the negative charge of the electrons would everywhere balance the positive charge of the various types of nuclei. This allowed him to reduce the problem to finding just the electron density as the density of the nuclei were assumed to be in constant proportion. We do not make this assumption, but we do make the strong assumption that our stellar object is made only of electrons and protons. This is necessary since the Thomas-Fermi model applies to fermions but not bosons. Although a hydrogen nucleus is a fermion, the other common nuclei like helium, carbon, or oxygen found in stellar objects are all bosons. Therefore we must make the restriction to only protons and electrons.   

The result is a two species generalization of Thomas-Fermi model which treats both the protons and electrons as fluids. This model can be used to describe the ground state of a white dwarf star. We will, however, not dwell on the physical motivation and justifications for this model: for a more detailed introduction to its motivation, see \cite{HK} or \cite{HK2}.

The two fluid model of course has similarities to the single fluid Thomas-Fermi model. But, because gravity will play an important role, it also shares similarities to models of self-gravitating and rotating fluids. Discussions of these closely related models can be found in \cite{AB1}, \cite{AB2}, \cite{FT}, \cite{LY}, and \cite{YYL}. For reference, the functional used in \cite{YYL} is given as 
\begin{equation}\label{rotatingFunctional}
    F(\rho)=\int_{\mathbb{R}^3}A(\rho(x))d^3x-\int_{\mathbb{R}^3}\rho(x) J(x)d^3x-\frac{1}{2}\int_{\mathbb{R}^3}\int_{\mathbb{R}^3}\frac{\rho(x)\rho(y)}{|x-y|}d^3xd^3y
\end{equation}
where $A(s)=s\int_0^sf(t)t^{-2}dt$, $f(\rho(x))=p(x)$ the equation of state when $p(x)$ is pressure, and $J$ is the angular velocity.   

The two fluid model is in some sense a combination of the single fluid Thomas-Fermi model and the self-gravitating fluid models, something we will again note in section \ref{themodel}, and in fact we will simply adapt some of the techniques used for the rotating models in our study.

The outline of the paper is as follows. In section \ref{themodel}, we will present and discuss the model. In section \ref{criticalpoints} we prove the existence of minimizers given certain bounds on the ratio of the particles. In section \ref{compactexistence} we present a different proof under stricter conditions giving existence with compactly supported minimizers. Section \ref{ELsection} presents an analysis of the Euler-Lagrange equations of the energy functional. Section \ref{synthesissection} relates the results of sections \ref{criticalpoints} and \ref{compactexistence} to the results of section \ref{ELsection}, and proves the uniqueness of the minimizers. Finally, we conclude in section \ref{conclusion}.  
\section{Energy Functionals}\label{themodel}
\subsection{The non-relativistic model}
Chandrasekhar's white dwarf model assumes a special relativistic kinetic energy. But before we generalize that model to two species, let us first consider the non-relativistic case; it already presents enough mathematical difficulties to make it interesting. The energy functional for the two species non-relativistic model is 
\begin{align}\label{functional}
    E(\rho_e, \rho_p)&=\frac{3h^2}{40m_e}\left(\frac{3}{\pi}\right)^{2/3}\int_{\mathbb{R}^3}\rho^{5/3}_e(x)d^3x+\frac{3h^2}{40m_p}\left(\frac{3}{\pi}\right)^{2/3}\int_{\mathbb{R}^3}\rho^{5/3}_p(x)d^3x\\
    & +\frac{q^2}{2}\int_{\mathbb{R}^3}\int_{\mathbb{R}^3}\frac{(\rho_p(x)-\rho_e(x))(\rho_p(y)-\rho_e(y))}{|x-y|}d^3yd^3x \nonumber\\
    &-\frac{G}{2}\int_{\mathbb{R}^3}\int_{\mathbb{R}^3}\frac{(m_p\rho_p(x)+m_e\rho_e(x))(m_p\rho_p(y)+m_e\rho_e(y))}{|x-y|}d^3yd^3x \nonumber
\end{align}
where $h$ is Planck's constant, $m_e$ is the electron mass, $m_p$ is the proton mass, $G$ is Newton's gravitation constant, and $q$ is the elementary charge. As a notational convenience, we will define $k_e:=\frac{3 h^2}{40m_e}\left(\frac{3}{\pi}\right)^{2/3}$ and $k_p:=\frac{3 h^2}{40m_p}\left(\frac{3}{\pi}\right)^{2/3}$. Mathematically, the exact values of these constants will not make much difference in what follows, and one could ignore any physical meaning mentioned in these introductory sections. We will however retain the suggestive notation and try to keep computed quantities in a form so that we can more easily make analogies to well known physical quantities for those readers familiar with, for example, the Chandrasekhar limiting mass. Although we have included no numerical results in this paper, it should be noted that for there are some computational issues with using the true physical values; a detailed discussion of this issue as well as numerical results can be found in \cite{HK}. 

As mentioned above, the first two terms on the right side of (\ref{functional}) are kinetic energy terms. These are derived from a phase space argument combined with the Pauli exclusion principle and nonrelativistic kinetic energies per particle, $\frac{p^2}{2m}$. A careful explanation of these integrals can be found in \cite{LS} or Chapter X of \cite{Chandra}. This is the term of the energy functional which does not hold for bosons: note that the Pauli exclusion principle was used in its justification. The third and fourth integrals are readily seen to be the electric and gravitational energies. 

It is also easy to see the similarities between this model and (\ref{rotatingFunctional}) and (\ref{TFequation}), and why we might say that this model is, mathematically at least, in some way a combination of those two models. From (\ref{TFequation}) we have taken the repulsive potential and from (\ref{rotatingFunctional}) we have taken the attractive potential. 

We wish to minimize $E$ subject to the following constraints: $\rho_e, \rho_p\geq0$, $\rho_e, \rho_p\in L^1\cap L^{5/3}$ (which is sufficient to make the third and fourth integrals finite), $\int\rho_e=N_e$, and $\int\rho_p=N_p$. We will sometimes refer to $E$ as $E_{(N_e, N_p)}$ when we want to make a point about particular values of $N_e$ and $N_p$, but for the most part we will only use $E$. As mentioned in the introduction, we are considering applications to stellar structures, so we have in mind $N_e, N_p$ on the order of $10^{56}$ (higher would make our nonrelativistic energy assumption questionable), although it makes no difference mathematically. Proofs that this functional is well defined and is bounded below can be found in, for example, \cite{YYL}. 

We can compute the Euler-Lagrange equations with $\int\rho_e=N_e$, and $\int\rho_p=N_p$ constraints to be
\begin{equation}\label{el1}
    \frac{5}{3}k_p\rho_p^{2/3}(x)+q^2\int_{\mathbb{R}^3}\frac{(\rho_p-\rho_e)(y)}{|x-y|}d^3y-Gm_p \int_{\mathbb{R}^3}\frac{(m_p\rho_p+m_e\rho_e)(y)}{|x-y|}d^3y=\lambda_p
\end{equation}
and
\begin{equation}\label{el2}
    \frac{5}{3}k_e\rho_e^{2/3}(x)- q^2\int_{\mathbb{R}^3}\frac{(\rho_p-\rho_e)(y)}{|x-y|}d^3y-Gm_e \int_{\mathbb{R}^3}\frac{(m_p\rho_p+m_e\rho_e)(y)}{|x-y|}d^3y=\lambda_e
\end{equation}
These will hold where $\rho_p>0$ and $\rho_e> 0$ respectively. If the densities were sufficiently regular, we could apply $-\Delta$ to (\ref{el1}) and (\ref{el2}) to get 
\begin{equation}\label{del1}
    -\frac{5}{3}k_p\Delta\rho^{2/3}_p+4\pi  q^2(\rho_p-\rho_e)-4\pi Gm_p(m_p\rho_p+m_e\rho_e)=0
\end{equation}
and 
\begin{equation}\label{del2}
    -\frac{5}{3}k_e\Delta\rho^{2/3}_e-4\pi  q^2(\rho_p-\rho_e)-4\pi Gm_e(m_p\rho_p+m_e\rho_e)=0
\end{equation}
again valid where $\rho_p>0$ and $\rho_e> 0$. Note that these are the equations of hydrostatic equilibrium for a self-gravitating, charged two fluid model if one use a polytropic equation of state of index 3/2. Since both functions are positive on their support, we define $\nu_{e,p}=\rho^{2/3}_{e,p}$. Then we rewrite (\ref{del1}) and (\ref{del2}) as 
\begin{equation}\label{nel1}
    \frac{5}{3}k_p\Delta\nu_p=4\pi  q^2(\nu^{3/2}_p-\nu^{3/2}_e)-4\pi Gm_p(m_p\nu^{3/2}_p+m_e\nu^{3/2}_e)
\end{equation}
and 
\begin{equation}\label{nel2}
    \frac{5}{3}k_e\Delta\nu_e=-4\pi  q^2(\nu^{3/2}_p-\nu^{3/2}_e)-4\pi Gm_e(m_p\nu^{3/2}_p+m_e\nu^{3/2}_e)
\end{equation}
So, aside from the question of regularity, we have reduced finding the critical points of (\ref{functional}) to finding solutions to this nonlinear elliptic system.

Positive solutions to systems of this form have been studied extensively; see \cite{DG} for an overview. The common techniques to find solutions are usually extensions of the techniques applied to find positive solutions of the scalar versions. For example, \cite{CFM1} generalizes the results of \cite{FLN}, \cite{Zou} generalizes \cite{GS1} and \cite{GS2}, and \cite{CFM2} generalizes \cite{BT}. These techniques cannot, at least straightforwardly, be applied to our system. One of the main issues is that many of the techniques use a fixed point iteration on a function space of positive functions. In our system, we cannot force the right hand sides of (\ref{nel1}) and (\ref{nel2}) to be positive, preventing us from applying the maximum principle in a crucial way.

But the more serious issue is that these techniques are simply not addressing the question we are interested in. These papers are overwhelmingly concerned with finding strictly positive solutions with Dirichlet boundary conditions on an arbitrary domain. While for spherical domains there do exist solutions satisfying these criterion, in fact easily describable as we will show below, we are also interested in solutions where one density has support strictly contained in the support of the other density. As we will show, there are only special pairs $(N_e, N_p)$ such that the densities have the same support. 


\subsection{The Special Relativistic Model}
Now we may turn to the direct generalization of the Chandrasekhar model. As mentioned above, the kinetic energy terms $k_e\int\rho_e^{5/3}$ and $k_p\int\rho_p^{5/3}$ are derived by assuming a nonrelativistic kinetic energy per particle of $\frac{p^2}{2m}$. If instead one uses the special-relativistic kinetic energy of $mc^2((1+p^2/2m)^{1/2}-1)$, one gets kinetic energy terms of the form 
\begin{equation}\label{oldrelenergy}
    \frac{\pi m_f^4c^5}{3h^3}\int_{\mathbb{R}^3}A\left(\frac{h}{ m_f c}\left(\frac{3}{8\pi}\right)^{1/3}\rho_f^{1/3}(x)\right)d^3x
\end{equation}
where
\begin{equation}\label{chandraformula}
    A(z)=8z^3[(z^2+1)^{1/2}-1]-z(2z^2-3)(z^2+1)^{1/2}-3\sinh^{-1}(z)
\end{equation}
and $f=e$ or $p$, see \cite{Chandra} Chapter X for the derivation of exactly this expression. This is however not always the best form of the kinetic energy to work with. Sometimes a better form is 
\begin{equation}\label{relenergy}
    m_fc^2\int_{\mathbb{R}^3}\int_0^{\rho_f(x)}\sqrt{1+\left(\frac{3}{\pi}\right)^{2/3}\left(\frac{h}{2m_fc}\right)^2\theta^{2/3}}d\theta d^3x.
\end{equation}
So to express the special relativistic energy functional, we only need to replace the first two terms in (\ref{functional}) with their analogues from (\ref{relenergy}). We will call this energy functional $E^S$. Then we make the same assumption of spherical symmetry and seek to minimize $E^S$ over the same set as before. If we compute the Euler-Lagrange equations of $E^S$ and assume high enough regularity to apply $-\Delta$, we find as analogues to (\ref{del1}) and (\ref{del2})
\begin{equation}\label{specialel1}
    -m_pc^2\Delta\sqrt{1+\left(\frac{3}{\pi}\right)^{2/3}\left(\frac{h}{2m_pc}\right)^2\rho_p^{2/3}}+4\pi  q^2(\rho_p-\rho_e)-4\pi Gm_p(m_p\rho_p+m_e\rho_e)=0
\end{equation}
and 
\begin{equation}\label{specialel2}
     -m_ec^2\Delta\sqrt{1+\left(\frac{3}{\pi}\right)^{2/3}\left(\frac{h}{2m_ec}\right)^2\rho_e^{2/3}}-4\pi  q^2(\rho_p-\rho_e)-4\pi Gm_e(m_p\rho_p+m_e\rho_e)=0.
\end{equation}
If, following \cite{Chandra} Chapter XI, we set $y_f^2=1+\left(\frac{3}{\pi}\right)^{2/3}\left(\frac{h}{2m_fc}\right)^2\rho_f^{2/3}$ where $f$ is $p$ or $e$, these equations can be rewritten as 
\begin{equation}\label{specialel3}
    \Delta y_p=\frac{4\pi}{m_pc^2}(q^2-Gm_p^2)\left(\frac{\pi}{3}\right)\left(\frac{2m_pc}{h}\right)^3(y_p^2-1)^{3/2}-\frac{4\pi}{m_pc^2}(q^2+Gm_pm_e)\left(\frac{\pi}{3}\right)\left(\frac{2m_ec}{h}\right)^3(y_q^2-1)^{3/2}
\end{equation}
and
\begin{equation}\label{specialel4}
    \Delta y_e=-\frac{4\pi}{m_ec^2}(q^2+Gm_pm_e)\left(\frac{\pi}{3}\right)\left(\frac{2m_pc}{h}\right)^3(y_p^2-1)^{3/2}+\frac{4\pi}{m_ec^2}(q^2-Gm_q^2)\left(\frac{\pi}{3}\right)\left(\frac{2m_ec}{h}\right)^3(y_q^2-1)^{3/2},
\end{equation}
these equations valid where $\rho_p>0$ and $\rho_e>0$, respectively. These are clearly much easier to work with than (\ref{specialel1}) and (\ref{specialel2}). In his single species model, Chandrasekhar is able to go a step further and introduce a change of variables to scale out the constants and essentially normalize the equation. He is left with an equation (assuming spherical symmetry) of the form 
\begin{equation}\label{chandradiffeq}
    \frac{1}{r^2}\frac{d}{dr}\left(r^2\frac{dy}{dr}\right)=-\left(y^2-\frac{1}{y_0^2}\right)^{3/2}
\end{equation}
where $y_0$ is a function of the central density. Rewriting the equation this way allows him to write certain quantities, such as the total mass of the configuration, with a simple expression. This is what allows him to write his well known limiting mass of possible configurations in a closed form.

Since we are working with a system, there is no scaling which will eliminate all the constants of (\ref{specialel3}) and (\ref{specialel4}) to get a system analogous to (\ref{chandradiffeq}). We therefore must work with (\ref{specialel3}) and (\ref{specialel4}) and will not focus on the question of limiting mass, although a few of our results will shed some light on it. 

\subsection{Zero-gravity limit}
Before we turn to proving the existence of minimizers for (\ref{functional}) and its special relativistic form, it should be noted that, although the gravitational force is very small in comparison to the electric force, setting $G=0$ trivializes the problem, at least in the non-relativistic case. Define 
\begin{align}\label{zeroG}
     E^0(\rho_e, \rho_p)&=k_e\int_{\mathbb{R}^3}\rho^{5/3}_e(x)d^3x+k_p\int_{\mathbb{R}^3}\rho^{5/3}_p(x)d^3x\\
    & +\frac{q^2}{2}\int_{\mathbb{R}^3}\int_{\mathbb{R}^3}\frac{(\rho_p(x)-\rho_e(x))(\rho_p(y)-\rho_e(y))}{|x-y|}d^3yd^3x \nonumber
\end{align}
We show that there is only a single local critical point, which happens to be the global minimum. The first two terms of $E^0$  are clearly positive, and the proof that the last term is positive can be found in Theorem II.6 in \cite{LS}. Therefore, $E^0$ has a global minimum of zero, obtained only when $\rho_{e}$ and $\rho_p$ are identically zero. But beyond this, for any fixed $N_p$ and $N_e$, $E^0$ has an infimum of zero. To see this, consider any $\rho_p$ and $\rho_e$ in the set of admissible functions. Then let $\rho^\lambda_e(s)=\rho_e(s/\lambda)/\lambda^3$ and $\rho^\lambda_p(s)=\rho_p(s/\lambda)/\lambda^3$. We have 
\begin{equation}
\int_{\mathbb{R}^3}\rho_p^\lambda(s)d^3s=\int_{\mathbb{R}^3}\frac{\rho_p(s/\lambda)}{\lambda^3}d^3s=\int_{\mathbb{R}^3}\rho_p(s)d^3s=N_p
\end{equation}
so $\rho_p^\lambda$ is admissible for $0<\lambda<\infty$. We also have 
\begin{equation}
\int_{\mathbb{R}^3}(\rho_p^\lambda)^{5/3}(s)d^3s=\int_{\mathbb{R}^3}\frac{\rho_p^{5/3}(s/\lambda)}{\lambda^5}d^3s=\int_{\mathbb{R}^3}\frac{\rho_p^{5/3}(s)}{\lambda^2}d^3s\rightarrow 0 \text{ as }\lambda\rightarrow\infty
\end{equation}
These statements are also true for $\rho_e$. We also have for the third term in $E^0$:
\begin{align}
    \int_{\mathbb{R}^3}\int_{\mathbb{R}^3}\frac{(\rho^\lambda_p(s)-\rho^\lambda_t(s))(\rho^\lambda_p(t)-\rho^\lambda_e(t))}{|s-t|}d^3sd^3t&=\int_{\mathbb{R}^3}\int_{\mathbb{R}^3}\frac{(\rho_p(s/\lambda)-\rho_t(s\lambda))(\rho_p(t/\lambda)-\rho_e(t/\lambda))}{|s-t|\lambda^6}d^3sd^3t\\
    &=\int_{\mathbb{R}^3}\int_{\mathbb{R}^3}\frac{(\rho_p(s)-\rho_t(s))(\rho_p(t)-\rho_e(t))}{|s-t|\lambda}d^3sd^3t \\
    &\rightarrow 0 \text{ as }\lambda\rightarrow\infty \nonumber
\end{align}
This process is physically analogous to spreading any given density out. We can then conclude that the infimum of $E^0$ for any fixed pair $(N_e, N_p)$ is zero, and not obtained. 

Beyond this, $E^0$ is also jointly convex, so the global minimum is the only critical point. To see this note that the first two terms of $E^0$ are strictly convex on $[0,\infty]$ since $x^{5/3}$ is. For the last term, consider 
\begin{equation}
\hat{E}^0(\rho_p, \rho_e)=\int_{\mathbb{R}^3}\int_{\mathbb{R}^3}\frac{(\rho_p(s)-\rho_t(s))(\rho_p(t)-\rho_e(t))}{|s-t|}d^3sd^3t
\end{equation}
Let $u_e$ and $u_p$ be admissible functions and define $\hat{E}^0(\alpha)=\hat{E}_G(\alpha\rho_p+(1-\alpha)u_p,\alpha\rho_e+(1-\alpha)u_e )$. Then we can compute 
\begin{equation}
\frac{d^2\hat{E}^0}{d\alpha^2}|_{\alpha=0}=\int_{\mathbb{R}^3}\int_{\mathbb{R}^3}\frac{(u_p(s)-u_t(s))(u_p(t)-u_e(t))}{|s-t|}d^3sd^3t
\end{equation}
As was stated above, this quantity is positive and only zero if $u_e=u_p$, so $E^0$ is jointly convex. So one might have considered using a perturbative expansion in $G$ around $G=0$ to approach finding solutions. But since $E^0$ has only trivial critical points this does not simplify the analysis; see \cite{KNY} for this approach.


\section{Minimizers of the Energy Functional}\label{criticalpoints}
This section generally follows the outline of \cite{Lions}.
\begin{theorem}\label{minexist}
If 
\begin{equation}\label{5/3bounds}
  \frac{1-\frac{Gm_p^2}{q^2}}{1+\frac{Gm_em_p}{q^2}}\leq\frac{N_e}{N_p}\leq\frac{1+\frac{Gm_em_p}{q^2}}{1-\frac{Gm_q^2}{q^2}},  
\end{equation}
there is a pair $(\rho_e, \rho_p)$ which minimizes $E$ such that $\rho_e, \rho_p\geq0$, $\rho_e, \rho_p\in L^1\cap L^{5/3}$, $\int\rho_e=N_e$, $\int\rho_p=N_p$, and $\rho_e, \rho_p$ are both radially symmetric.
\end{theorem}

Define $W^{N_e, N_p}$ to be the set of pairs of functions $(\rho_e, \rho_p)$ such that $\rho_f\geq0$, $\rho_f\in L^1\cap L^{5/3}$, $\rho_f$ are radially symmetric, and $\int\rho_f\leq N_f$. Also, let $\mathcal{E}^{N_e, N_p}$ be the minimum energy of $E$ over $W^{N_e, N_p}$; this could possibly be $-\infty$. To avoid the clunky notation, we will always fix $N_e$ and $N_p$ and just write $W$ and $\mathcal{E}$.

 Sometimes we will also use the notation
\begin{equation}
    B\rho(x):=\int_{\mathbb{R}^3}\frac{\rho(y)}{|x-y|}d^3y.
\end{equation}
First, we want to show
\begin{lemma}\label{Negenergy}
For any $(N_e, N_p)$ satisfying $(\ref{5/3bounds})$, $\mathcal{E}<0$.
\end{lemma}
\begin{proof}
It is easy to see that this is true in the case that $N_e=N_p$. For if we take any configuration in $W$ such that $\rho_e=\rho_p$, the electric energy is zero and we are left with 
\begin{equation}\label{reducedfunctional}
    E(\rho_e, \rho_p)=(k_e+k_p)\int_{\mathbb{R}^3}\rho^{5/3}_e(x)d^3x
    -\frac{G}{2}(m_e+m_p)^2\int_{\mathbb{R}^3}\int_{\mathbb{R}^3}\frac{\rho_e(x)\rho_e(y)}{|x-y|}d^3yd^3x
\end{equation}
Using the same ``spreading out'' argument we used when we considered the zero-gravity case, we see that as $\lambda\rightarrow\infty$, the kinetic energy terms go to zero like $1/\lambda^2$ while the gravitational energy goes to zero like $1/\lambda$. Then for large enough $\lambda$, $E(\rho_e^\lambda, \rho_p^\lambda)<0$.

Therefore, because each of the terms in $E$ is in some sense continuous in $(\rho_e, \rho_p)$ (without trying to make this precise in any way, we just use the intuition), there is some interval of values for $N_e/N_p$ containing 1 such that $\mathcal{E}<0$. We determine the size of this interval by finding example functions in $W$.

Suppose that for a fixed $(N_e, N_p)$ such that the ratio $N_e/N_p$ satisfies $(\ref{5/3bounds})$, we have a pair $(\rho_e, \rho_p)$ both with compact support such that $E(\rho_e, \rho_p)< 0$. We claim we can add a small positive or negative charge outside of this configuration and the energy will not increase. Formally, the idea is to see how many ``test particles'' we can bring in from infinity and still have negative energy. Of course, this is not a rigorous argument, so we now make this idea rigorous. 

Let us assume we are adding positive charge. Since our space $W$ consists of spherically symmetric functions, we must add our charge in a spherically symmetric configuration. We may assume the configuration to which we are adding has compact support, suppose it is contained in a ball of radius $R_0$. Then add to $(\rho_e, \rho_p)$ a particle density described by the function $h=\epsilon\chi_S$ for $S=\{x| R_0<|x|<R_0+\eta\}$. Then we compute
\begin{align}\label{energydiff}
    E(\rho_e, \rho_p+g)-E(\rho_e, \rho_p)&=k_p\epsilon^{5/3}\frac{4}{3}\pi\left[(R_0+\eta)^3-R_0^3\right]+\left[\frac{q^2}{2}-\frac{G}{2}\right]\epsilon^2\int_{S}\int_{S}\frac{1}{|x-y|}d^3yd^3x \nonumber\\
    &+\epsilon\int_{S}\int_{\mathbb{R}^3}\frac{q^2(\rho_p-\rho_e)-Gm_p(m_p\rho_p+m_e\rho_e)}{|x-y|}\\
    &=k_p\epsilon^{5/3}\frac{4}{3}\pi\left[3R_0^2\eta+3R_0\eta^2+\eta^3\right] \nonumber\\
    &+\left[\frac{q^2}{2}-\frac{G}{2}\right]\epsilon^2 4\pi\int_{R_0}^{R_0+\eta}r^2\left[\frac{4}{3}\frac{\pi}{r}\left[r^3-R_0^3\right]+2\pi\left[(R_0+\eta)^2-r^2\right]\right]dr\nonumber \\
    &+4\pi \epsilon \int_{R_0}^{R_0+\eta} \left[q^2(N_p-N_e)-Gm_p(m_pN_p+m_eN_e)\right]rdr\\
    &=k_p\epsilon^{5/3}\frac{4}{3}\pi\left[3R_0^2\eta+3R_0\eta^2+\eta^3\right]\nonumber\\&+\left[\frac{q^2}{2}-\frac{G}{2}\right]\epsilon^2 4\pi\left[4\pi R_0^3\eta^2+\frac{16\pi}{3}R_0^2\eta^3+\frac{8\pi}{3}R_0\eta^4+\frac{3\pi}{5}\eta^5\right]\nonumber \\
    &-4\pi\epsilon N_e(q^2-Gm_p^2)\left(B_p-\frac{N_p}{N_e}\right)(2R_0\eta+\eta^2)
\end{align}
where $B_p$ is the multiplicative inverse of the left side of (\ref{5/3bounds}). If we assume $\left(B_p-\frac{N_p}{N_e}\right)>0$, the last term is strictly negative. Therefore, we may take $\epsilon>0$ small enough so that the difference is negative. So we can add some positive amount of charge to the configuration and the energy will become more negative. The same will be true of adding a small negative amount of charge using the strict version of the right half of (\ref{5/3bounds}). 

We intend to treat this as an iterative procedure, and note that if $N_e=N_p$, we can assume that the densities have compact support; the size of the support makes no difference for the argument that $\mathcal{E}^{N_e, N_e}<0$. The issue remaining is whether we can saturate (\ref{5/3bounds}) using this procedure.

To show that we can, we let $\eta=\epsilon$. Then if we have only added a finite amount of charge to the configuration, we can also be assured the resulting configuration has finite radius. So no matter how much charge we add, we can always apply the procedure again as long as the strict version of (\ref{5/3bounds}) is satisfied. This will give us (\ref{5/3bounds}). So we need to show that the procedure can be carried out if $\eta=\epsilon$, which is to say that we can find an $\epsilon$ so that in this case (\ref{energydiff}) is negative.

The leading order in $\epsilon$ parts of the three terms of (\ref{energydiff}) behave as $\epsilon^{8/3}R_0^2$, $\epsilon^3R_0^4$, and $\epsilon^2 R_0$, respectively. So if we take $\epsilon$ to be $o(R_0^{-4})$, we can make the first two terms smaller in absolute value than the last, negative term. 
\end{proof}

Let 
\begin{equation}
    u_f^n(x)=\int_{\mathbb{R}^3}\frac{\rho_f^n(y)}{|x-y|}d^3y
\end{equation}
for $f=e,p$. Then due to the uniform bound on the $L^{4/3}$ norm, we have that $\nabla u_f^n$ is bounded in $L^q$ for $3/2<q\leq 12/5$ and $u_f^n$ is bounded in $L^p$ for $3<p\leq 12$, see \cite{AB1} proposition 6. Note also that $u_f^n$ is radially symmetric. Since $\rho_f^n$ is uniformly bounded in the $L^{4/3}$ norm, we can apply Banach-Alaoglu and find a weakly convergent subsequence. From this subsequence, extract a further subsequence converging weakly in $L^{6/5}$. Let $(\rho_e, \rho_p)$ be the functions to which $(\rho_e^n, \rho_p^n)$ converge. We can extract yet a further subsequence so that $(u_e^n, u_p^n)$ converge weakly in $L^6$ and a.e. to $(u_e, u_p)$.


The goal is to prove that $(\rho_e, \rho_p)$ is a minimizer of $E$ over $W$ and $\int\rho_p=N_p$ and $\int\rho_e=N_e$. To see the former, we show that $E$ is weakly lower semicontinuous. First note that
\begin{equation}\label{LSCcompress}
   H(g):= \int g^{5/3}\leq \liminf \int g_n^{5/3}
\end{equation}
for the same reasons as given in \cite{Lions} (in that paper, instead of $g^{5/3}$ there is given a general $j(\rho)$ nonnegative, continuous, and convex). So we only need to focus on proving that the two other terms in the energy are lower semicontinuous. This result follows by proving that both $(u_p^n)$ and $(u_e^n)$ are relatively compact in the $(L^{6/5})^*=L^6$ topology and so some subsequence converges to $(u_e, u_p)$. Coupled with the weak $L^{6/5}$ convergence of $\rho_f^n$, we then have
\begin{equation}\label{LSCgrav}
\int_{\mathbb{R}^3}(m_p\rho_p^n(x)+m_e\rho_e^n(x))(m_pu_p^n(x)+m_eu_e^n(x))d^3x\rightarrow \int_{\mathbb{R}^3}(m_p\rho_p(x)+m_e\rho_e(x))(m_pu_p(x)+m_eu_e(x))d^3x,
\end{equation}
and 
\begin{equation}\label{LSCelec}
\int_{\mathbb{R}^3}(\rho_p^n(x)-\rho_e^n(x))(u_p^n(x)-u_e^n(x))d^3x\rightarrow \int_{\mathbb{R}^3}(\rho_p(x)-\rho_e(x))(u_p(x)-u_e(x))d^3x.
\end{equation}

To see that $(u_f^n)$ is relatively compact in $L^6$, we first apply Proposition II.1 from \cite{Lions} to obtain the bound
\begin{equation}
    |u_f^n(x)|\leq C(N,p,q)\left(||\nabla u_f^n||_{L^p}^{p'/q+p'}|| u_f^n||_{L^q}^{q/q+p'}\right)|x|^{-2p'/(q+p')} 
\end{equation}
where $p'$ is the H\"older conjugate of of $p$ and $C$ is a positive constant. This allows us to conclude that $u_f^n\rightarrow 0$ as $|x|\rightarrow \infty$ uniformly in $x$ and $n$. This is requirement v) of Compactness Lemma 2 from \cite{Strauss}. Since we already have a.e. convergence of $u_f^n$, it is straightforward to check that with $P^n(s)=s^6$ and $Q^n(s)=s^7+\chi_{B_1}(s)s^5$ where $\chi$ is an indicator function, the other four requirements are also satisfied. We can then apply the lemma to conclude that 
\begin{equation}
    \int_{\mathbb{R}^3}||u_f^n(x)|^6-|u_f(x)|^6|d^3x\rightarrow0
\end{equation}
This gives strong convergence of $(u_f^n)$ to $u_f$ in $L^6$.


Using (\ref{LSCcompress}), (\ref{LSCgrav}), and (\ref{LSCelec}), we can conclude that $E$ is weakly lower semicontinuous, and therefore 
\begin{equation}
    E(\rho_e, \rho_p)\leq \liminf E(\rho_e^n, \rho_p^n)
\end{equation}
So we have a minimizer $(\rho_e, \rho_p)$. However, we can only so far conclude that that $\int\rho_p\leq N_p$ and $\int\rho_e\leq N_e$; we need to prove these are actual equalities. 

We will use the following well known lemma, which can be found in, for example, \cite{AB1} as proposition 5:
\begin{lemma}\label{AB5}
Suppose $\rho\in L^1\cap L^p$. If $1<p\leq 3/2$, then $B\rho$ is in $L^r$ for $3<r<3p/(3-2p)$, and 
\begin{equation}
    ||B\rho||_r\leq c_0(||\rho||_1^b||\rho||_p^{1-b}+||\rho||_1^c||\rho||_p^{1-c})
\end{equation}
where $0<b<1$ and $0<c<1$.
\end{lemma}

\begin{lemma}\label{rightmass}
The minimizer $(\rho_e, \rho_p)$ satisfies $\int\rho_p= N_p$ and $\int\rho_e=N_e$.
\end{lemma}
\begin{proof}
First, assume that both $\int\rho_e< N_e$ and $\int \rho_p<N_p$. We can argue in a way similar to \cite{Lions}. Let $\rho_f^\sigma(x)=\rho_f(\frac{x}{\sigma^{1/3}})$. Then \begin{equation}
\int_{\mathbb{R}^3}\rho_f^\sigma(x)d^3x=\sigma\int_{\mathbb{R}^3}\rho(x)d^3x
\end{equation}
and 
\begin{align}
    E(\rho_e^\sigma, \rho_p^\sigma)&=\sigma\left(k_e\int_{\mathbb{R}^3}\rho^{5/3}_ed^3x+k_p\int_{\mathbb{R}^3}\rho^{5/3}_pd^3x\right)\\
    &+\frac{q^2\sigma^{5/3}}{2}\int_{\mathbb{R}^3}\int_{\mathbb{R}^3}\frac{(\rho_p(x)-\rho_e(x))(\rho_p(y)-\rho_e(y))}{|x-y|}d^3yd^3x \nonumber\\
    &-\frac{G\sigma^{5/3}}{2}\int_{\mathbb{R}^3}\int_{\mathbb{R}^3}\frac{(m_p\rho_p(x)+m_e\rho_e(x))(m_p\rho_p(y)+m_e\rho_e(y))}{|x-y|}d^3yd^3x \nonumber
\end{align}
Then we can compute 
\begin{equation}
    \frac{dE}{d\sigma}(\rho_e^\sigma, \rho_p^\sigma)|_{\sigma=1}=K(\rho_e, \rho_p)+\frac{5}{3}V(\rho_e,\rho_p)=E(\rho_e,\rho_p)+\frac{2}{3}V(\rho_e,\rho_p)
\end{equation}
where we have called the kinetic energy terms $K$ and the potential energy terms $V$, so that $E(\rho_e,\rho_p)=K(\rho_e,\rho_p)+V(\rho_e,\rho_p)$. But we know that $E(\rho_e,\rho_p)<0$, and since $K(\rho_e,\rho_p)>0$, we must then have $V(\rho_e,\rho_p)<0$. Therefore, 
\begin{equation}
    \frac{dE}{d\sigma}(\rho_e^\sigma, \rho_p^\sigma)|_{\sigma=1}<0,
\end{equation}
so we can increase $\sigma$ slightly so that we still have $\int \rho_e\leq N_e$ and $\int\rho_p\leq N_p$ but the energy is decreased. This contradicts $(\rho_e, \rho_p)$ being a minimizer.

This leaves us with the two cases such $\int\rho_e=N_e$ while $\int\rho_p<N_p$ and $\int\rho_e<N_e$ while $\int\rho_p=N_p$. These cases can be treated in the same way, so let us consider the former. If both densities have bounded support, then we are in the same situation as we considered in lemma (\ref{Negenergy}): adding a small sphere of positive charge outside of both supports will actually decrease the energy of the configuration. Then we contradict that $(\rho_e,\rho_p)$ is a minimizer. 

If either or both supports are unbounded, define $N_f(R)$ to be the number of particles of type $f$ within a sphere of radius $R$. Then the function $F(R)=q^2(N_p(R)-N_e(R))-Gm_p(m_pN_p(R)+m_eN_e(R))$ is continuous, and approaches some negative value as $R\rightarrow \infty$. Let us choose $R_0$ so that $F(R)<\beta<0$ for $R>R_0$. Now we add a spherical shell of positive charge of width $\eta$, letting $h$ be as in lemma \ref{Negenergy},  and compute the change in energy. We find that 
\begin{align}
    &E(\rho_e, \rho_p+h)-E(\rho_e, \rho_p)= k_p\int_{S}((\rho_p+\epsilon)^{5/3}-\rho_p^{5/3})d^3x\\
    +& \left[\frac{q^2}{2}-\frac{G}{2}\right]\epsilon^2 4\pi\left[4\pi R_0^3\eta^2+\frac{16\pi}{3}R_0^2\eta^3+\frac{8\pi}{3}R_0\eta^4+\frac{3\pi}{5}\eta^5\right]\nonumber\\
    +&4\pi\epsilon \int_{R_0}^{R_0+\eta}\left[q^2(N_p(r)-N_e(r))-Gm_p(m_pN_p(r)+m_eN_e(r))\right]rdr \nonumber\\
    +&4\pi\epsilon\int_{R_0}^{R_0+\eta}r^2\int_{|y|>R_0+r}\frac{q^2(\rho_p-\rho_e)(y)-m_pG(m_p\rho_p+m_e\rho_e)(y)}{|y|}d^3ydr \nonumber
\end{align}
The two differences from the situation we considered before are that now we must be concerned with the kinetic energy term, and the final term, corresponding to the interaction between the sphere of added particles and the configuration outside this sphere. We want to show that the third term on the right side is larger in absolute value than the other terms. By our assumption on the size of $R_0$, we can say that this term is less than $-4\pi\epsilon\eta R_0\beta $

To leading order in $\epsilon$, the kinetic energy term is $\pi k_p\frac{20}{3}\epsilon\int_{R_0}^{R_0+\eta}\rho_p^{2/3}(r)dr$. Because $\int \rho_p=N_p$, we must be able to find a set of total width in the radial direction of $\eta$ and with radial distance at least $R_0$ such that $\rho_p(r)\leq\frac{1}{r^3}$. We will assume for simplicity that this set is $S$, the spherical shell, but it may be that $S$ needs to be broken up into many unconnected shells. Then the kinetic energy term is less than $\pi k_p\frac{20}{3}\epsilon \frac{\eta}{R_0}$, clearly smaller than the third term for large enough $R_0$. 

The second term is smaller than the third by taking both $\epsilon$ and $\eta$ small. For the last term, since both $\rho_p$ and $\rho_e$ are integrable, one can choose $R_0$ large enough so that $\int_{|y|>R_0}q^2(\rho_p+\rho_e)(y)<\hat{\epsilon}$ for any chosen $\hat{\epsilon}$. Then we can say that the last term is bounded by $\frac{4\pi\epsilon}{3R_0}(3R_0\eta^2+3\eta R^2_0+\eta^3)\hat{\epsilon}$. The leading order term in $\epsilon$ is then $4\pi\epsilon\eta R_0\hat{\epsilon}$. Compare this to the leading order term for the negative third term, $4\pi\epsilon^2 R_0\beta$. Then we make $R_0$ at least large enough so that $\hat{\epsilon}<|\beta|$, and then choose $\epsilon$ so that the negative third term dominates the positive terms. 

This contradicts that we have found a minimizer, and we find that we have $\int\rho_e=N_e$ and $\int\rho_p=N_p$. Note that we have assumed continuity of the minimizer for simplicity. This will independently be proven below.
\end{proof}

This completes the proof of the theorem.

\subsection{The Special Relativistic Case}
The same proof also works for the special relativistic case with a few minor changes which we will discuss below. However, recall that in the special relativistic case we do not know that the solutions are decreasing, so our result is not as strong. Here we can redefine $W$ so that it contains functions in $L^1\cap L^{4/3}$ instead of $L^1\cap L^{5/3}$. We then have
\begin{theorem}\label{specialcase}
If 
\begin{equation}\label{SRbounds}
  \frac{1-\frac{Gm_p^2}{q^2}}{1+\frac{Gm_em_p}{q^2}}\leq\frac{N_e}{N_p}\leq\frac{1+\frac{Gm_em_p}{q^2}}{1-\frac{Gm_q^2}{q^2}},  
\end{equation}
and $[m_pN_p+m_eN_e]^{2/3}<\frac{\pi 2^{2/3}hc\left(\frac{3}{8\pi}\right)^{4/3}}{GKm_p^{4/3}}$, then there is a pair $(\rho_e,\rho_p)$ which minimizes $E^s$ in $W$.
\end{theorem}
The constant $K$ here is the same constant which appears in \cite{Lions} equation (5'). To save space, let $\alpha_f$ denote $\frac{h}{m_f c}\left(\frac{3}{8\pi}\right)^{1/3}$.
\begin{enumerate}
    \item In the proof that the energy of the minimizer is negative if (\ref{SRbounds}) is satisfied, the first step was to consider the $N_p=N_e$ case. The ``spreading out'' argument does not work as nicely since $g(z)$ is not a homogeneous polynomial. But if we take a $\rho_f$ which is bounded and again say that $\rho_f^\lambda(x)=\rho_f(x/\lambda)/\lambda^3$, we can say that, since $g(z)$ behaves like $z^5$ as $z\rightarrow 0$, we have 
    \begin{align}
        \int_{\mathbb{R}^3} g(\alpha_f (\rho_f^\lambda(x))^{1/3})d^3x&=\int_{\mathbb{R}^3} g\left(\frac{\alpha_f}{\lambda} \left(\rho_f\left(\frac{x}{\lambda}\right)\right)^{1/3}\right)d^3x\\
        &\leq \int_{\mathbb{R}^3} M\left(\frac{\alpha_f}{\lambda}\right)^{9/2}\rho_f\left(\frac{x}{\lambda}\right)^{3/2}d^3x\\
        &=\frac{1}{\lambda^{3/2}}\int_{\mathbb{R}^3} M\alpha_f^{9/2}\rho_f(x)^{3/2}d^3x
    \end{align}
    for large $\lambda$ and some $M$. It therefore decays in $\lambda$ faster than the gravitational energy, and we can find a negative energy state. 
    
    We can use the same idea with the second part of the argument: since the argument is that we can shrink the added density until the gravitational energy dominates, the kinetic energy term will again behave like $\epsilon^{5/3}$ for small $\epsilon$.
    
    \item The process of obtaining the uniform $L^{4/3}$ bounds on $\rho^n_f$ must be changed to add the extra condition from the statement of theorem \ref{specialcase}. Before, we took the bound on the gravitational energy, 
    \begin{equation}
        C_1=2^{1/3}GKm_p^{4/3}(m_pN_p+m_eN_e)^{2/3},
    \end{equation}
    and found a lower bound on $\int\rho^{5/3}$ involving $(C_1+\epsilon)\int\rho^{4/3}$. This was to bound the negative gravitational energy in terms of the positive kinetic energy. We were able to do this regardless of what the value of $C_1$ was because $\lim (\rho^{5/3})/\rho^{4/3}=\infty$. But this is no longer true with the special relativistic kinetic energy. Instead, we have (this asymptotic can be found in \cite{Chandra} chapter X)
    \begin{equation}
        \lim_{\rho\rightarrow\infty}\frac{g\left(\alpha_f \rho_f^{1/3}\right)}{\rho^{4/3}}=6\left[\left(\frac{h}{ m_f c}\right)^4\left(\frac{3}{8\pi}\right)^{4/3}\right],
    \end{equation}
    so if we have 
    \begin{equation}\label{massbounds}
        (m_pN_p+m_eN_e)^{2/3}<\frac{\pi2^{2/3}h c\left(\frac{3}{8\pi}\right)^{4/3}}{(GKm_p^{4/3})}
    \end{equation}
    Then we can apply the same procedure as in the nonrelativistic case to bound the $||\rho_f||_{L^{4/3}}$ norms. 
    \item Because $g(z)\geq 0$ for $z\geq 0$, we can again apply Fatou's lemma to conclude that 
    \begin{equation}
    \int g(\alpha_f \rho_f^{1/3})\leq \liminf \int g(\alpha_f\rho^n_f)^{1/3})
    \end{equation}
    \item For proving that $\int\rho_p=N_p$ and $\int\rho_e=N_e$, the case that $\int\rho_p<N_p$ and $\int\rho_e<N_e$ is exactly the same. The case that either $\int\rho_p<N_p$ while $\int\rho_e=N_e$ or $\int\rho_p=N_p$ and $\int\rho_e<N_e$ can be adapted using the same arguments as given by item 1. on this list.
    
    \end{enumerate}

\subsection{Critical mass estimate}
In the special relativistic case, it is not too difficult to prove something of a converse to Theorem
\ref{specialcase}.  This is the existence of a mass limit analogous to the Chandrasekhar limit found in the effective singular density case \cite{Chandra}. 

Fix $(N_e, N_p)$. If we take $\rho_f=C_f\chi_{B_R}$ for $C_f=\frac{3N_f}{4\pi R^3}=:\beta_f^3/R^3$, then we can get an exact expression for the energy. The electric and gravitational energies simplify to $(3/10R)(q^2-GM^2)$, where $Q$ is the total charge of the configuration and $M$ is the total mass of the configuration. The kinetic energies we compute as the unappealing 
\begin{equation}
    \frac{4}{3}\pi R^3\sum_f\frac{\pi m^4_fc^5}{3h^3}\left[8\alpha_f^3\frac{\beta_f^3}{R^3}\left(\left(\alpha_f^2\frac{\beta_f^2}{R^2}+1\right)^{1/2}-1\right)-\alpha_f\frac{\beta_f}{R}\left(2\alpha_f^2\frac{\beta_f^2}{R^2}-3\right)\left(\alpha_f^2\frac{\beta_f^2}{R^2}+1\right)^{1/2}-\sinh^{-1}\left(\alpha_f\frac{\beta_f}{R}\right)\right].
\end{equation}
It is however, easy to see that this term behaves like $K/R$ for some constant $K$ when $R$ is very small, see again \cite{Chandra} chapter X for the asymptotic expression. Therefore, if we let $R\rightarrow0$, we see that the potential and kinetic energies are proportional. By using that $\lim_{z\rightarrow\infty} A(z)\sim 6z^4$, we can compute that 
\begin{equation}
K=\frac{3^{2/3}}{2^{10/3}}\frac{hc}{\pi^{2/3}}(N_e^{4/3}+N_p^{4/3}).
\end{equation}
By comparing $K$ to $(3/10)(q^2-GM^2)$, we can see that, if the gravitational potential dominates the electric and kinetic energies, letting $R\rightarrow 0$ gives us a sequence of configurations with energy unbounded below. Then for such combinations of $N_e$ and $N_p$, no minimizer could exist. This result is analogous to the Chandrasekhar Mass of the single density model:
\begin{equation}
    -\left(\frac{hc}{G}\right)^{3/2}\frac{2^{7/2}}{3^{1/2}\pi}\left(\xi^2\frac{d\theta_3}{d\xi}\right)_{\xi=\xi_3}\frac{1}{(\mu_eH)^2} 
\end{equation}
where $\theta_3$ is the solution to the normalized Lane-Emden equation of index 3, $\xi_3$ is the first zero of this solution, $\mu_e$ is the average molecular weight per electron, and $H$ is the mass of the hydrogen atom. The last fraction in particular shows how this value depends on the composition of the star; in our case that would correspond to depending on $N_p/N_e$. In fact, if we write $N_p/N_e=\frak{r}$, we can compute that the energy of configurations with this ratio is unbounded below if 
\begin{equation}\label{crudeestimate}
   \left(\frac{hc}{G}\right)^{3/2}\frac{5^{3/2}}{2^{7/2}3^{1/2}\pi}\frac{(\frak{r}^{4/3}+1)^{3/2}}{\left[(m_p\frak{r}+m_e)^2-\frac{C}{G}q^2(\frak{r}-1)^2\right]^{3/2}}<N_e
\end{equation}
This tells us that for each admissible ratio $N_e/N_p$, like in the Chandrasekhar case, there is some upper bound on the total mass of the configuration, a crude estimate of which is given by the left hand side of (\ref{crudeestimate}) multiplied by $(\frak{r}m_p+m_e)$.
\section{Minimizers of the Energy Functional with Compact Support}\label{compactexistence}

This section presents a proof of the existence of radially symmetric minimizers of (\ref{functional}) and its special relativistic version. It is a use of the general technique outlined in sections 5 and 6 of \cite{AB1} which prove the existence of minimizers of functionals of the form (\ref{rotatingFunctional}), subject to certain constraints on $A$, $\rho$, and $J$. However, there is an important difference in the result of this section as compared with the previous. In the process of proving existence, we also prove that there is a minimizer with compact support. This comes at the cost of losing the equalities in (\ref{5/3bounds}); that is, the technique cannot show that when the ratio $N_e/N_p$ equals either of the bounds in (\ref{5/3bounds}) a compactly supported minimizer exists. This is because the unique minimizer does not have compact support when $N_p/N_e$ saturates either of the bounds in ($\ref{5/3bounds}$). 

For $R>0$ and fixed $(N_e, N_p)$, let $W_R$ be the set of pairs $(\rho_e, \rho_p)$ such that $\int \rho_f=N_f$, $\rho_f\geq0$, $\rho_f$ is radially symmetric, $\rho_f=0$ outside the ball of radius $R$, and such that $\rho_f\leq R$; note that unlike in the previous section, these functions integrate to $N_f$ exactly. This set will be nonempty for large enough $R$. Define $W$ to be as $W_R$ without the $R$ related constraints. 
\begin{lemma}\label{compactminimizer}
For $R$ large enough, there is a pair $(\rho^R_e, \rho^R_p)\in W_R$ which minimizes $E_{(N_e, N_p)}$ on $W_R$.
\end{lemma}
\begin{proof}
It is straightforward to see that $W_R$ is convex, and, fixing $1<p<\infty$, that $W_R$ is a closed and bounded subset of $L^p\times L^p$. Since any bounded and closed subset of $L^p\times L^p$ is weakly compact by Banach-Alaoglu, we show that $E_{(N_e, N_p)}$ is lower semicontinuous in the weak topology for $p=\frac{5}{3}$. 

Since we are considering functions on a bounded set, it is clear that if $\rho_{e,n}$ converges weakly in $L^{5/3}$ to $\rho_e$, we have 
\begin{equation}
    \int_{\mathbb{R}^3}\rho^{5/3}_{e,n}d^3x\rightarrow \int_{\mathbb{R}^3}\rho^{5/3}_{e}d^3x.
\end{equation}
The same is true for the second term in the functional.

For the third term, note that the operator $$B:\rho\rightarrow \int_{\mathbb{R}^3}\frac{\rho(y)}{|x-y|}d^3y$$ is a compact operator from $L^{5/3}$ to $L^\infty$, by Sobolev embedding. Suppose that $\{(\rho_{e,n}, \rho_{p,n})\}$ converges weakly to $(\rho_e,\rho_p)$. The sequence $\{\rho_{p,n}-\rho_{e,n}\}$ is thus bounded. Therefore $B(\rho_{p,n}-\rho_{e,n})\rightarrow B(\rho_{p}-\rho_{e})$ in $L^\infty$. Then we can say
\begin{align}
    &\int_{\mathbb{R^3}}|(\rho_{p,n}-\rho_{e,n})(x)B(\rho_{p,n}-\rho_{e,n})(x)-(\rho_{p}-\rho_{e})(x)B(\rho_{p}-\rho_{e})(x)|d^3x \nonumber \\
    &\leq \int_{\mathbb{R^3}} |\rho_{p,n}(x)-\rho_{e,n}(x)||B(\rho_{p,n}-\rho_{e,n})(x)-B(\rho_{p}-\rho_{e})(x)|d^3x \\
    &+\int_{\mathbb{R^3}} |(\rho_{p,n}(x)-\rho_{e,n}(x))-(\rho_{p}-\rho_{e})(x))||B(\rho_{p}-\rho_{e})(x)|d^3x \nonumber\\
    &\rightarrow 0 \nonumber
\end{align}
So the third and also fourth terms of the functional are even continuous in the weak topology. 

We may therefore conclude that if $\{(\rho_{e,n}, \rho_{p,n})\}$ is a minimizing sequence in $W_R$ which converges weakly to $(\rho_{e}, \rho_{p})$, $$E_{(N_e, N_p)}(\rho_{e}, \rho_{p})\leq \liminf E_{(N_e, N_p)}(\rho_{e, n}, \rho_{p,n})$$ so $(\rho_e,\rho_p)$ is a minimizer, which we now call $(\rho^R_e,\rho^R_p)$.
\end{proof}

The next step, again following \cite{AB1}, is to show that there is some $\hat{R}$ such that for $R\geq \hat{R}$, the minimizer in $W_R$ lies in $W_{\hat{R}}$. Then a limiting argument can be applied to show this function will be the minimizer when we consider functions on the entire space. To prove this, we use (\ref{el1}) and (\ref{el2}), but restricted to $W_R$, to say that 
\begin{equation}\label{AB191}
    \frac{5}{3}k_p(\rho_p^R)^{2/3}(x)= -q^2\int_{\mathbb{R}^3}\frac{(\rho_p^R-\rho_e^R)(y)}{|x-y|}d^3y+Gm_p \int_{\mathbb{R}^3}\frac{(m_p\rho_p^R+m_e\rho_e^R)(y)}{|x-y|}d^3y+\lambda^R_p
\end{equation}
where $\rho_p^R>0$, and 
\begin{equation}\label{AB192}
    \frac{5}{3}k_e(\rho_e^R)^{2/3}(x)=  q^2\int_{\mathbb{R}^3}\frac{(\rho_p^R-\rho_e^R)(y)}{|x-y|}d^3y+Gm_e \int_{\mathbb{R}^3}\frac{(m_p\rho_p^R+m_e\rho_e^R)(y)}{|x-y|}d^3y+\lambda^R_e
\end{equation}
where $\rho_e^R>0$. We will show that, under certain conditions, the right sides of (\ref{AB191}) and (\ref{AB192}) become strictly negative for large $R$, so we must have that $\rho_{e}^R$ and $\rho_{p}^R$ vanish for large $R$. 

We will also need that, on the sets where $\rho_p^R=0$ and $\rho_e^R=0$, we have 
\begin{equation}\label{minimizerInequality1}
    \lambda_p^R\leq q^2\int_{\mathbb{R}^3}\frac{(\rho_p^R-\rho_e^R)(y)}{|x-y|}d^3y-Gm_p \int_{\mathbb{R}^3}\frac{(m_p\rho_p^R+m_e\rho_e^R)(y)}{|x-y|}d^3y
\end{equation}
and
\begin{equation}\label{minimizerInequality2}
    \lambda_e^R\leq -q^2\int_{\mathbb{R}^3}\frac{(\rho_p^R-\rho_e^R)(y)}{|x-y|}d^3y-Gm_e \int_{\mathbb{R}^3}\frac{(m_p\rho_p^R+m_e\rho_e^R)(y)}{|x-y|}d^3y,
\end{equation}
respectively. See \cite{AB1} for details.

We have 
\begin{theorem}
If 
\begin{equation}
  \frac{1-\frac{Gm_p^2}{q^2}}{1+\frac{Gm_em_p}{q^2}}<\frac{N_e}{N_p}<\frac{1+\frac{Gm_em_p}{q^2}}{1-\frac{Gm_q^2}{q^2}},  
\end{equation}
 there is a pair $(\rho_e, \rho_p)$ which minimizes (\ref{functional}) in $W$. Moreover, both $\rho_e$ and $\rho_p$ have compact support. 
\end{theorem}
This theorem is proved using several lemmas.
\begin{lemma}\label{43bounds}
There is a constant $k_0$ such that $||\rho_e^R||_{4/3},||\rho_p^R||_{4/3}\leq k_0$ for $R\geq R_0$.
\end{lemma}
\begin{proof}
First note that $E(\rho_e^R, \rho_p^R)\leq E(\rho_e^{R_0}, \rho_p^{R_0}) $ when $R\geq R_0$.
From (\ref{functional}), we get two inequalities
\begin{equation}\label{lemma41}
    k_e\int_{\mathbb{R}^3}(\rho_e^R)^{5/3}(x)d^3x\leq \frac{G}{2}\int_{\mathbb{R}^3}\int_{\mathbb{R}^3}\frac{(m_p\rho^R_p(x)+m_e\rho^R_e(x))(m_p\rho^R_p(y)+m_e\rho^R_e(y))}{|x-y|}d^3yd^3x+E(\rho_e^{R_0}, \rho_p^{R_0})
\end{equation}
and 
\begin{equation}\label{lemma42}
    k_p\int_{\mathbb{R}^3}(\rho_p^R)^{5/3}(x)d^3x\leq \frac{G}{2}\int_{\mathbb{R}^3}\int_{\mathbb{R}^3}\frac{(m_p\rho^R_p(x)+m_e\rho^R_e(x))(m_p\rho^R_p(y)+m_e\rho^R_e(y))}{|x-y|}d^3yd^3x+E(\rho_e^{R_0}, \rho_p^{R_0})
\end{equation}
We apply proposition 6 of \cite{AB1} to conclude 
\begin{align}
    &\frac{G}{2}\int_{\mathbb{R}^3}\int_{\mathbb{R}^3}\frac{(m_p\rho^R_p(x)+m_e\rho^R_e(x))(m_p\rho^R_p(y)+m_e\rho^R_e(y))}{|x-y|}d^3yd^3x \nonumber\\
    &\leq \hat{c}\left(\int_{\mathbb{R^3}}|m_p\rho^R_p(x)+m_e\rho^R_e(x)|^{4/3}d^3x\right)(m_eN_e+m_pN_p)^{2/3}\\
    &\leq \hat{c}||\rho_e^R||^{4/3}_{4/3}+\hat{c}||\rho_p^R||^{4/3}_{4/3}
\end{align}
where we have absorbed constants into $\hat{c}$. There are three cases to consider. In the first, both $||\rho_e^R||_{4/3}$ and $||\rho_p^R||_{4/3}$ are bounded as $R\rightarrow\infty$, so there is nothing to prove. In the second, one of $||\rho_e^R||_{4/3}$ or $||\rho_p^R||_{4/3}$ is bounded while the other is unbounded. In the third case, both are unbounded.

Let us consider the second case, and suppose $||\rho_e^R||_{4/3}$ is unbounded. For any $\hat{c}$, there is some value $s_\epsilon$ such that $$k_es^{5/3}\geq (\hat{c}+2\epsilon)s^{4/3}$$ when $s>s_\epsilon$. Therefore,
\begin{align}
    \int_{\mathbb{R}^3}|\rho_e^R(x)|^{4/3}d^3x&=\int_{\rho_e^R<s_\epsilon}|\rho_e^R(x)|^{4/3}d^3x+\int_{\rho_e^R>s_\epsilon}|\rho_e^R(x)|^{4/3}d^3x\\
    &\leq N_es_\epsilon^{1/3}+(\hat{c}+2\epsilon)^{-1}k_e\int_{\mathbb{R}^3}|\rho_e^R(x)|^{5/3}d^3x
\end{align}
For $||\rho_e^R||_{4/3}^{4/3}$ sufficiently large, this implies (note that $2\epsilon$ is replaced with $\epsilon$)
\begin{equation}\label{lemma43}
(\hat{c}+\epsilon)\int_{\mathbb{R}^3}|\rho_e^R(x)|^{4/3}d^3x\leq k_e\int_{\mathbb{R}^3}|\rho_e^R(x)|^{5/3}d^3x
\end{equation}
Plugging (\ref{lemma43}) into (\ref{lemma41}), we get 
\begin{equation}
    \epsilon||\rho_e^R||_{4/3}^{4/3}\leq E(\rho_e^{R_0}, \rho_p^{R_0})+\hat{c}||\rho_p^R||_{4/3}^{4/3}
\end{equation}
which is a contradiction if $||\rho_e^R||_{4/3}$ is unbounded and $||\rho_p^R||_{4/3}$ is bounded. The case with the roles of the densities reversed is of course the same.

Now let us consider the case when both $||\rho_p^R||_{4/3}$ and $||\rho_e^R||_{4/3}$ are unbounded. We will apply the same sort of ideas as in the previous case, but slightly modified. Since both norms are unbounded, we can use the same argument to show that for $R$ large enough, 
\begin{equation}\label{lemma44}
    \left(\hat{c}+\frac{\hat{c}}{1-\epsilon}\right)\int_{\mathbb{R}^3}|\rho_e^R(x)|^{4/3}d^3x\leq k_e\int_{\mathbb{R}^3}|\rho_e^R(x)|^{5/3}d^3x
\end{equation}
and 
\begin{equation}\label{lemma45}
    2\hat{c}\int_{\mathbb{R}^3}|\rho_p^R(x)|^{4/3}d^3x\leq k_p\int_{\mathbb{R}^3}|\rho_p^R(x)|^{5/3}d^3x
\end{equation}
Using these inequalities with (\ref{lemma41}) and (\ref{lemma42}), we get 
\begin{equation}
    \frac{\hat{c}}{1-\epsilon}||\rho_e^R||_{4/3}^{4/3}\leq E(\rho_e^{R_0}, \rho_p^{R_0})+\hat{c}||\rho_p^R||_{4/3}^{4/3}
\end{equation}
and 
\begin{equation}
    \hat{c}||\rho_p^R||_{4/3}^{4/3}\leq E(\rho_e^{R_0}, \rho_p^{R_0})+\hat{c}||\rho_e^R||_{4/3}^{4/3}
\end{equation}
which again gives a contradiction. This completes the proof.
\end{proof}
\begin{lemma}\label{easyBounds}
$\lambda_e^R<0$ and $\lambda_p^R<0$ for all $R$.
\end{lemma}
\begin{proof}
Fix $R$ and consider $|x|=2R$. Since $\rho_e^R(x)=\rho_p^R(x)=0$, we have that $\lambda_p^R$ and $\lambda_e^R$ satisfy (\ref{minimizerInequality1}) and (\ref{minimizerInequality2}), respectively. Since our densities have spherical symmetry, we can apply Newton's theorem to these integrals to get 
\begin{equation}
    \lambda_p^R\leq q^2\frac{N_p-N_e}{2R}-Gm_p\frac{m_pN_p+m_eN_e}{2R}
\end{equation}
and 
\begin{equation}
    \lambda_e^R\leq -q^2\frac{N_p-N_e}{2R}-Gm_e\frac{m_pN_p+m_eN_e}{2R}
\end{equation}
The result follows from (\ref{5/3bounds}).
\end{proof}
\begin{lemma}\label{linfinitybound}
$||\rho_e^R||_\infty, ||\rho_p^R||_\infty<k_1$ for all $R>R_0$.
\end{lemma}
\begin{proof}
 We combine (\ref{AB191}) and (\ref{AB192}) with the results of Lemma (\ref{easyBounds}). Since $\lambda_e^R<0$ and $\lambda_p^R<0$, we have 
 \begin{equation}\label{maxbound1}
     \frac{5}{3}k_p(\rho_p^R)^{2/3}(x)\leq \max\{0, -q^2\int_{\mathbb{R}^3}\frac{(\rho_p^R-\rho_e^R)(y)}{|x-y|}d^3y+Gm_p \int_{\mathbb{R}^3}\frac{(m_p\rho_p^R+m_e\rho_e^R)(y)}{|x-y|}d^3y\}
 \end{equation}
 and 
  \begin{equation}\label{maxbound2}
     \frac{5}{3}k_e(\rho_e^R)^{2/3}(x)\leq \max\{0, q^2\int_{\mathbb{R}^3}\frac{(\rho_p^R-\rho_e^R)(y)}{|x-y|}d^3y+Gm_e \int_{\mathbb{R}^3}\frac{(m_p\rho_p^R+m_e\rho_e^R)(y)}{|x-y|}d^3y\}
 \end{equation}
 as the densities cannot be negative. Since we have shown that $||\rho_f||_{4/3}$ is bounded, we apply proposition 5 with $p=4/3$ from \cite{AB1} to the right hand side of (\ref{maxbound1}) to conclude (choosing $r=6$ for concreteness)
 \begin{align}\label{NewtonBound}
     &||B(-q^2(\rho_p^R-\rho_p^R)+Gm_p(m_p\rho_p^R+m_e\rho_e^R))||_6 \nonumber\\
     &\leq c_0||(-q^2(\rho_p^R-\rho_p^R)+Gm_p(m_p\rho_p^R+m_e\rho_e^R)||_1^{b}||(-q^2(\rho_p^R-\rho_p^R)+Gm_p(m_p\rho_p^R+m_e\rho_e^R)||_{4/3}^{(1-b)}\\
     &+c_0||(-q^2(\rho_p^R-\rho_p^R)+Gm_p(m_p\rho_p^R+m_e\rho_e^R)||_1^{c}||(-q^2(\rho_p^R-\rho_p^R)+Gm_p(m_p\rho_p^R+m_e\rho_e^R)||_{4/3}^{1-c} \nonumber\\
     &\leq c_0\left[q^2(N_p+N_e)+Gm_p(m_pN_P+m_eN_e)\right]^{b}\left[2q^2k_0^2+2Gm_p^2k_0^2\right]^{1-b}\\
     &+c_0\left[q^2(N_p+N_e)+Gm_p(m_pN_P+m_eN_e)\right]^{c}\left[2q^2k_0^2+2Gm_p^2k_0^2\right]^{1-c}. \nonumber
 \end{align}
 So we can conclude that the $L^6$ norms of these potentials are uniformly bounded for all $R$ large enough; the same result holds for (\ref{maxbound2}).
 
 From here we use a bootstrapping procedure. The results of the previous paragraph combined with (\ref{maxbound1}) and (\ref{maxbound2}) imply that both $\rho_p^R$ and $\rho_e^R$ are uniformly bounded in the $L^4$ norm. We then apply the second conclusion of proposition 5 in \cite{AB1} and the same inequalities as above but with $4/3$ replaced with $4$ to conclude that the right hand sides of (\ref{maxbound1}) and (\ref{maxbound2}) are uniformly bounded for all $R$ in $L^\infty$ norm. But this implies that $||\rho_e^R||_\infty$ and $||\rho_p^R||_\infty$ are also both uniformly bounded.
\end{proof}
\begin{lemma}\label{energyBounds}
There is an $e<0$ such that $E(\rho_e^R, \rho_p^R)\leq e$ for all large $R$.
\end{lemma}
\begin{proof}
It is easy to see that this is true in the case that $N_e=N_p$. For if we take any configuration in $W$ such that $\rho_e=\rho_p$, the electric energy is zero and we are left with 
\begin{equation}\label{reducedfunctional}
    E(\rho_e, \rho_p)=(k_e+k_p)\int_{\mathbb{R}^3}\rho^{5/3}_e(x)d^3x
    -\frac{G}{2}(m_e+m_p)^2\int_{\mathbb{R}^3}\int_{\mathbb{R}^3}\frac{\rho_e(x)\rho_e(y)}{|x-y|}d^3yd^3x
\end{equation}
Using the same ``spreading out'' argument we used when we considered the zero-gravity case, we see that as $\lambda\rightarrow\infty$, the kinetic energy terms go to zero like $1/\lambda^2$ while the gravitational energy goes to zero like $1/\lambda$. Then for large enough $\lambda$, $E(\rho_e^\lambda, \rho_p^\lambda)<0$. Since this procedure works for any neutral configuration, it works for one of compact support, and so we must take $R$ larger than this $\lambda$ if we want to find a configuration with negative energy in $W_R$.

Therefore, because each of the terms in (\ref{functional}) is in some sense continuous in $(\rho_e, \rho_p)$ (without trying to make this precise in any way, we just use the intuition), there is some interval of values for $N_e/N_p$ containing 1 such that $\mathcal{E}<0$. We determine the size of this interval by finding example functions in $W$.

Suppose that for a fixed $(N_e, N_p)$ such that the ratio $N_e/N_p$ satisfies $(\ref{5/3bounds})$, we have a pair $(\rho_e, \rho_p)$ both with compact support such that $E(\rho_e, \rho_p)< 0$. We claim we can add a small positive or negative charge outside of this configuration and the energy will not increase. Formally, the idea is to see how many ``test particles'' we can bring in from infinity and still have negative energy. Of course, this is not a rigorous argument, so we now make this idea rigorous. 

Let us assume we are adding positive charge. Since our space $W_R$ consists of spherically symmetric functions, we must add our charge in a spherically symmetric configuration. The configuration to which we are adding has compact support, suppose it is contained in a ball of radius $R_0$. Then add to $(\rho_e, \rho_p)$ a particle density described by the function $h=\epsilon\chi_S$ for $S=\{x| R_0<|x|<R_0+\eta\}$. Then we compute
\begin{align}\label{energydiff}
    E(\rho_e, \rho_p+g)-E(\rho_e, \rho_p)&=k_p\epsilon^{5/3}\frac{4}{3}\pi\left[(R_0+\eta)^3-R_0^3\right]+\left[\frac{q^2}{2}-\frac{G}{2}\right]\epsilon^2\int_{S}\int_{S}\frac{1}{|x-y|}d^3yd^3x \nonumber\\
    &+\epsilon\int_{S}\int_{\mathbb{R}^3}\frac{q^2(\rho_p-\rho_e)-Gm_p(m_p\rho_p+m_e\rho_e)}{|x-y|}\\
    &=k_p\epsilon^{5/3}\frac{4}{3}\pi\left[3R_0^2\eta+3R_0\eta^2+\eta^3\right] \nonumber\\
    &+\left[\frac{q^2}{2}-\frac{G}{2}\right]\epsilon^2 4\pi\int_{R_0}^{R_0+\eta}r^2\left[\frac{4}{3}\frac{\pi}{r}\left[r^3-R_0^3\right]+2\pi\left[(R_0+\eta)^2-r^2\right]\right]dr\nonumber \\
    &+4\pi \epsilon \int_{R_0}^{R_0+\eta} \left[q^2(N_p-N_e)-Gm_p(m_pN_p+m_eN_e)\right]rdr\\
    &=k_p\epsilon^{5/3}\frac{4}{3}\pi\left[3R_0^2\eta+3R_0\eta^2+\eta^3\right]\nonumber\\&+\left[\frac{q^2}{2}-\frac{G}{2}\right]\epsilon^2 4\pi\left[4\pi R_0^3\eta^2+\frac{16\pi}{3}R_0^2\eta^3+\frac{8\pi}{3}R_0\eta^4+\frac{3\pi}{5}\eta^5\right]\nonumber \\
    &-4\pi\epsilon N_e(q^2-Gm_p^2)\left(B_p-\frac{N_p}{N_e}\right)(2R_0\eta+\eta^2)
\end{align}
where $B_p$ is the multiplicative inverse of the left side of (\ref{5/3bounds}). If we assume $\left(B_p-\frac{N_p}{N_e}\right)>0$, the last term is strictly negative. Therefore, we may take $\epsilon>0$ small enough so that the difference is negative. So we can add some positive amount of charge to the configuration and the energy will become more negative. The same will be true of adding a small negative amount of charge using the strict version of the right half of (\ref{5/3bounds}). 

We intend to treat this as an iterative procedure, and note that if $N_e=N_p$, we can assume that the densities have compact support; the size of the support makes no difference for the argument that $\mathcal{E}^{N_e, N_e}<0$. The issue remaining is whether we can saturate (\ref{5/3bounds}) using this procedure.

To show that we can, we let $\eta=\epsilon$. Then if we have only added a finite amount of charge to the configuration, we can also be assured the resulting configuration has finite radius. So no matter how much charge we add, we can always apply the procedure again as long as the strict version of (\ref{5/3bounds}) is satisfied. This will give us (\ref{5/3bounds}). So we need to show that the procedure can be carried out if $\eta=\epsilon$, which is to say that we can find an $\epsilon$ so that in this case (\ref{energydiff}) is negative.

The leading order in $\epsilon$ parts of the three terms of (\ref{energydiff}) behave as $\epsilon^{8/3}R_0^2$, $\epsilon^3R_0^4$, and $\epsilon^2 R_0$, respectively. So if we take $\epsilon$ to be $o(R_0^{-4})$, we can make the first two terms smaller in absolute value than the last, negative term. 
\end{proof}
\begin{lemma}
$||B(m_e\rho_e^R+m_p\rho_p^R)||_\infty\geq -\frac{2e}{G(m_pN_p+m_eN_p)}$ for all large $R$.
\end{lemma}
\begin{proof}
From the previous lemma, we have
\begin{align}
-e&\leq -E(\rho_e^R, \rho_p^R)\leq \frac{G}{2}\int_{\mathbb{R}^3}\int_{\mathbb{R}^3}\frac{(m_p\rho^R_p(x)+m_e\rho^R_e(x))(m_p\rho^R_p(y)+m_e\rho^R_e(y))}{|x-y|}d^3yd^3x\\
&\leq \frac{G}{2}(m_pN_p+m_eN_p)||B(m_e\rho_e^R+m_p\rho_p^R)||_\infty
\end{align}
where in the second line we have applied Holder's inequality. 
\end{proof}
\begin{lemma}\label{diffusion}
There is an $\epsilon>0$ such that for large $R$,
\begin{equation}
    \int_{|x-x_R|<1}\rho_f^R(x)d^3x\geq \epsilon
\end{equation}
for $f=e$ or $p$ and some $x_R$.
\end{lemma}
\begin{proof}
The point of this lemma is to show that the densities cannot spread out indefinitely; at least one must have some mass accumulated somewhere.

Define 
\begin{equation}
    \epsilon_e^R=\sup_x\int_{|y-x|<1}\rho_e^R(y)d^3y
\end{equation}
and $\epsilon_p^R$ analogously. We have 
\begin{align}
    B(m_p\rho_p^R+m_e\rho_e^R)(x)&=\int_{\mathbb{R}^3}\frac{(m_p\rho^R_p(y)+m_e\rho^R_e(y))}{|x-y|}d^3y\\
    &=\int_{|x-y|<1}\frac{(m_p\rho^R_p(y)+m_e\rho^R_e(y))}{|x-y|}d^3y+\int_{1<|x-y|<r}\frac{(m_p\rho^R_p(y)+m_e\rho^R_e(y))}{|x-y|}d^3y\\
    &+\int_{|x-y|>r}\frac{(m_p\rho^R_p(y)+m_e\rho^R_e(y))}{|x-y|}d^3y \nonumber
\end{align}
The last term is bounded by $\frac{m_pN_p+m_eN_e}{r}$. The set $1<|x-y|<r$ can be covered by $cr^3$ balls of radius 1, and can therefore be bounded by $cr^3(m_p\epsilon_p^R+m_e\epsilon_e^R)$. Since lemma (\ref{linfinitybound}) shows that $\rho_e^R$ and $\rho_p^R$ are bounded in the $L^\infty$ norm, we can apply proposition 5 with $p=r=\infty$ of \cite{AB1} on the ball of radius 1 to get 
\begin{equation}
    \int_{|x-y|<1}\frac{(m_p\rho^R_p(y)+m_e\rho^R_e(y))}{|x-y|}d^3y\leq k_4((m_e\epsilon_e^R+m_p\epsilon_p^R)^b+(m_e\epsilon_e^R+m_p\epsilon_p^R)^c)
\end{equation}
Combining these three estimates, we get 
\begin{equation}
    B(m_p\rho_p^R+m_e\rho_e^R)(x)\leq \frac{m_pN_p+m_eN_e}{r}+cr^3(m_p\epsilon_p^r+m_e\epsilon_e^R)+k_4((m_e\epsilon_e^R+m_p\epsilon_p^R)^b+(m_e\epsilon_e^R+m_p\epsilon_p^R)^c)
\end{equation}
Now if both $\epsilon_e^R$ and $\epsilon_p^R$ could get arbitrarily close to 0 as $R$ became large, then we could choose $r$ to be very large, and then $\epsilon_e^R$ and $\epsilon_p^R$ to be very small and we would get that 
\begin{equation}
    B(m_p\rho_p^R+m_e\rho_e^R)(x)<-\frac{2e}{G(m_pN_p+m_eN_p)}
\end{equation}
contradicting the previous lemma. 

Note that because of spherical symmetry, we also get a bound on $|x_R|$ for all $R$; we will say that $|x_R|<r_0$.
\end{proof}
\begin{lemma}
For some $l<0$, $\lambda_e^R+\lambda_p^R<l$ for large $R$. 
\end{lemma}
\begin{proof}
Suppose $R>>r_0$. If $|x-x_R|<r<R$, then we can say that
\begin{equation}
   \int_{\mathbb{R}^3}\frac{(m_p\rho^R_p(y)+m_e\rho^R_e(y))}{|x-y|}d^3y\geq\int_{|y-x_R|<1}\frac{(m_p\rho^R_p(y)+m_e\rho^R_e(y))}{|x-y|}d^3y\geq \frac{m_e\epsilon}{r+1}
\end{equation}
which is independent of $R$, but of course depends on $x$. As $|x-x_R|$ increases (and $r$ with it), we see that this value decreases as $r^{-1}$. Because $\int \rho^R_p=N_p$ and $\int \rho^R_e=N_e$, in the $r$-sphere around $x_R$ there must be some point $x$ such that $\rho_f^R(x)<kr^{-3}$, for $k$ independent of $R$.  So if we take $r$ large enough, we can find an $x$ such that the sum of the kinetic energy terms is less than half the negative of the gravitational term.

We want to find the same $x$ for both $\rho^R_e$ and $\rho^R_p$, so we prove that for any $r$ large enough, there is always a set $U_f$ such that $m(U_f)>\frac{2}{3}\pi r^3$, where $m$ means the measure,  and on $U_f$, $\rho_f^R<k_1r^{-3}$. Then $m(U_p\cap U_e)\neq 0$, so there must be some $\overline{x}$ such that $\rho_e^{R}(\overline{x})<k_1r^{-3}$ and $\rho_p^{R}(\overline{x})<k_1r^{-3}$. Assume this is not true. Then there is a set $S_f$ such that $m(S_f)\geq \frac{2}{3}\pi r^3$ and $\rho_f^R>k_1r^{-3}$ on $S_f$. For the correctly chosen $k_1$ (still independent of $R$) this is a contradiction. 

In the ``worst case'' scenario, at $\overline{x}$ both densities are positive and we add equations (\ref{AB191}) and (\ref{AB192}) together to obtain
\begin{equation}\label{lemma49}
    \lambda_e^R+\lambda_p^R=\frac{5}{3}k_p(\rho_p^R)^{2/3}(\overline{x})+\frac{5}{3}k_e(\rho_e^R)^{2/3}(\overline{x})-G(m_e+m_p)\int_{\mathbb{R}^3}\frac{(m_p\rho^R_p(y)+m_e\rho^R_e(y))}{|\overline{x}-y|}d^3y
\end{equation}
If one of the densities is 0, then (\ref{lemma49}) is not true. But this would mean that $|x|>|\overline{x}|$ implies $\rho_f^R(x)=0$. As we intend to use (\ref{lemma49}) to show that when $R$ is large one of the densities \textit{must} be zero for $|x|$ large, this is no problem. 

 So for any $R$ large enough, we have found an $\overline{x}$ such that the right side is negative, and therefore $\lambda_e^R+\lambda_p^R<l<0$.
\end{proof}
\begin{lemma}
For $R$ large enough, there is an $\alpha$ such that for $|x|>\alpha$, $B(m_p\rho^R_p+m_e\rho^R_e)(x)\leq -l$.
\end{lemma}
\begin{proof}
Define $\epsilon_\alpha^R$ for $\alpha<<R$ by 
\begin{equation}
\epsilon_\alpha^R=\sup_{|x|>\alpha/2}\int_{|x-y|<1}(m_e\rho_e^R+m_p\rho_p^R)(y)dy
\end{equation}
Radial symmetry implies that $\epsilon^R_\alpha$ decays faster than $1/\alpha$ as $\alpha$ gets larger. Then if we take $|x|>\alpha$, using the same arguments as in lemma \ref{diffusion}, we find
\begin{equation}
    B(m_p\rho_p^R+m_e\rho_e^R)(x)\leq c_1((\epsilon_\alpha^R)^a+(\epsilon_\alpha^R)^b)+c_1\epsilon_\alpha^Rr^3+c_3r^{-1}
\end{equation}
So we can take first $r$ very large, and then $\alpha$ very large to get that $B(m_p\rho^R_p+m_e\rho^R_e)(x)\leq -l$, which holds for all $R$ large enough.
\end{proof}
As long as both densities are nonzero, we can add (\ref{AB191}) and (\ref{AB192}) together and rewrite it as 
\begin{equation}
    \frac{5}{3}k_p(\rho_p^R)^{2/3}(x)+\frac{5}{3}k_e(\rho_e^R)^{2/3}(x) =G(m_e+m_p)\int_{\mathbb{R}^3}\frac{(m_p\rho^R_p(y)+m_e\rho^R_e(y))}{|x-y|}d^3y+\lambda_e^R+\lambda_p^R
\end{equation}
But by the last lemma, when $R$ is very large, for $|x|>\alpha$ the left side would be negative, so one of the densities must become zero at a smaller $|x|$ so that this is no longer an equality.  Then for $|x|$ large enough, we may assume that one of the densities is zero, although we do not know which one. The final step is to prove that there is a $\beta$ such that if $|x|>\beta$, both densities are 0. 

We need 
\begin{lemma}
Both densities have contiguous support.
\end{lemma}
\begin{proof}
Suppose that $\rho_e^R$ did not have contiguous support. By this we mean that there is some open set $A$ such that $\rho_e^R=0$ on $A$, but that there is another open set $O$ on which $\rho_e^R>0$ and such that 
\begin{equation}
    \sup\{|x|:x\in A\}\leq \inf \{|x|: x\in O\}
\end{equation}
Suppose $O$ is such that this is actually an equality ($O$ and $A$ are adjacent).

The idea is that we can move some amount of $\rho_e^R$ and lower the total energy contradicting that $(\rho_e^R, \rho_p^R)$ is a minimizer. Consider
\[ v_1(x)= \begin{cases} 
      -\rho_e^R & x\in K_O \\
      \rho_e^R(x+\alpha) &  x\in K_A \\
      0 & \text{else}
   \end{cases}
\]
where $K_A$ is some small set contained in $A$, and $K_O$ is $K_A$ shifted by the amount $\alpha$ so that it lies in $O$.
Then consider 
\begin{align}
    \delta_1&=E(\rho_e^R+v_1, \rho_p^R)-E(\rho_e^R, \rho_p^R)\\
    &= -q^2\int_{\mathbb{R}^3}\int_{\mathbb{R}^3}\frac{v_1(x)(\rho_p^R(y)-\rho_e^R(y))}{|x-y|}d^3y-Gm_e\int_{\mathbb{R}^3}\int_{\mathbb{R}^3}\frac{v_1(x)(m_p\rho_p^R(y)+m_e\rho_e^R(y))}{|x-y|}d^3y\\
    &+\frac{q^2}{2}\int_{\mathbb{R}^3}\int_{\mathbb{R}^3}\frac{v_1(x)v_1(y)}{|x-y|}d^3y-\frac{Gm^2_e}{2}\int_{\mathbb{R}^3}\int_{\mathbb{R}^3}\frac{v_1(x)v_1(y)}{|x-y|}d^3y \nonumber\\
    &=\int_{K_O}\rho_e^R(x)\int_{\mathbb{R}^3}\frac{q^2(\rho_p^R(y)-\rho_e^R(y))+Gm_e(m_p\rho_p^R(y)+m_e\rho_e^R(y))}{|x-y|}d^3y\\
    &-\int_{K_A}\rho_e^R(x+\alpha )\int_{\mathbb{R}^3}\frac{q^2(\rho_p^R(y)-\rho_e^R(y))+Gm_e(m_p\rho_p^R(y)+m_e\rho_e^R(y))}{|x-y|}d^3y \nonumber\\
    &+\frac{q^2}{2}\int_{\mathbb{R}^3}\int_{\mathbb{R}^3}\frac{v_1(x)v_1(y)}{|x-y|}d^3y-\frac{Gm^2_e}{2}\int_{\mathbb{R}^3}\int_{\mathbb{R}^3}\frac{v_1(x)v_1(y)}{|x-y|}d^3y \nonumber\\
\end{align}
As we can make the sets $K_A$ and $K_O$ as small as we like, the first two terms in the last equality can be made to dominate the last two, so let us focus on the first two. Now if $\delta_1<0$ for any choice of $K_A$ and $K_O$, then $(\rho_e^R, \rho_p^R)$ is not a minimizer after all. So assume that $\delta_1\geq0$.

We now split the proof into two cases. In the first, there is not an open set $D$ on which $\rho_e^R>0$ and such that 
\begin{equation}
    \sup\{|x|:x\in D\}\leq \inf \{|x|: x\in A\}
\end{equation}
This would imply that the $\rho_e^R=0$ at the origin. In case 2, there is such a $D$. We handle the cases separately. 

Define $V(x)=-B(q^2(\rho_p^R(y)-\rho_e^R(y))+Gm_e(m_p\rho_p^R(y)+m_e\rho_e^R(y)))$, so that $V$ is the potential felt by the electron density. In case 1, $\delta_1>0$ implies that $V(x)$ is smaller on $O$ than on $A$. But
\begin{equation}
    -\Delta V=-q^2\rho_p^R(y)-Gm_em_p\rho_p^R(y)\leq 0
\end{equation}
implies that $\Delta V\geq 0$ on $A$, or that the maximum of $V$ on $A$ occurs at its boundary. But $O$ is the boundary of $A$, so we get a contradiction.

In case 2, consider
\[ v_2(x)= \begin{cases} 
      -\rho_e^R & x\in K_D \\
      \rho_e^R(x-\beta) &  x\in K_A \\
      0 & \text{else}
   \end{cases}
\]
where $K_D$ is $K_A$ shifted by $-\beta$ and on which $\rho_e^R>0$. Further assume $D$ is on the boundary of $A$ (so that $O$ and $D$ are on the boundaries of $A$). Then consider 
\begin{align}
    \delta_2&=E(\rho_e^R+v_2, \rho_p^R)-E(\rho_e^R, \rho_p^R)\\
    &= -q^2\int_{\mathbb{R}^3}\int_{\mathbb{R}^3}\frac{v_2(x)(\rho_p^R(y)-\rho_e^R(y))}{|x-y|}d^3y-Gm_e\int_{\mathbb{R}^3}\int_{\mathbb{R}^3}\frac{v_2(x)(m_p\rho_p^R(y)+m_e\rho_e^R(y))}{|x-y|}d^3y\\
    &+\frac{q^2}{2}\int_{\mathbb{R}^3}\int_{\mathbb{R}^3}\frac{v_2(x)v_2(y)}{|x-y|}d^3y-\frac{Gm^2_e}{2}\int_{\mathbb{R}^3}\int_{\mathbb{R}^3}\frac{v_2(x)v_2(y)}{|x-y|}d^3y \nonumber\\
    &=\int_{K_D}\rho_e^R(x)\int_{\mathbb{R}^3}\frac{q^2(\rho_p^R(y)-\rho_e^R(y))+Gm_e(m_p\rho_p^R(y)+m_e\rho_e^R(y))}{|x-y|}d^3y\\
    &-\int_{K_A}\rho_e^R(x-\beta )\int_{\mathbb{R}^3}\frac{q^2(\rho_p^R(y)-\rho_e^R(y))+Gm_e(m_p\rho_p^R(y)+m_e\rho_e^R(y))}{|x-y|}d^3y \nonumber\\
    &+\frac{q^2}{2}\int_{\mathbb{R}^3}\int_{\mathbb{R}^3}\frac{v_2(x)v_2(y)}{|x-y|}d^3y-\frac{Gm^2_e}{2}\int_{\mathbb{R}^3}\int_{\mathbb{R}^3}\frac{v_2(x)v_2(y)}{|x-y|}d^3y \nonumber\\
\end{align}
Again, we have that if $\delta_2<0$, then $(\rho_e^R, \rho_p^R)$ is not a minimizer. So assume that $\delta_2\geq0$. If $\delta_1>0$ and $\delta_2>0$, then $V(x)$ is lesser on both $K_O$ and $K_D$ than on $K_A$. But again we have that $\Delta V\geq 0$ in $A$, or that the maximum for $V$ in $A$ is on the boundary. This is again a contradiction and one of $\delta_1$ or $\delta_2$ must be negative. Therefore the densities must be contiguous. 

In the case that $\delta_1$ or $\delta_2$ are equal to zero, $(\rho_e^R+v, \rho_p^R)$ is also a minimizer, but one where we have moved some electron density into $Z_n$. If $\rho_e^R$ is still not contiguous, we apply the procedure again and try to find where $\delta_1$ and $\delta_2$ are both greater than zero. So we repeat until either we find such $\delta_1$ and $\delta_2$, or have a contiguous density. 

The case in which $\rho_p^R=0$ is similar. 
\end{proof}

So not only can we assume that for $|x|$ large enough one of the densities is zero, we must have that it is the same density as we increase $|x|$.

Suppose it is $\rho_e^R$, and let $L_R=\sup\{|x||\rho_e^R(x)>0\}$. Now if we take $|x|>L_R$, then we find that (assuming $\rho_p^R\neq 0$)
\begin{align}
    \lambda_p^R&=\frac{5}{3}k_p(\rho_p^R)^{2/3}(x) +q^2\int_{\mathbb{R}^3}\frac{(\rho_p^R-\rho_e^R)(y)}{|x-y|}d^3y-Gm_p \int_{\mathbb{R}^3}\frac{(m_p\rho_p^R+m_e\rho_e^R)(y)}{|x-y|}d^3y\\
    &=\frac{5}{3}k_p(\rho_p^R)^{2/3}(x)+\frac{q^2(N_p(|x|)-N_e)}{|x|}-\frac{Gm_p(m_pN_p(|x|)+m_eN_e)}{|x|}\\
    &+4\pi[q^2-Gm_p^2]\int_{|x|}^R \rho^R_p(r)rdr \nonumber \\
    &\leq \frac{5}{3}k_p(\rho_p^R)^{2/3}(x)+\frac{q^2(N_p-N_e)}{|x|}-\frac{Gm_p(m_pN_p+m_eN_e)}{|x|}
\end{align}
where by $N_p(|x|)$ we mean the number of positively charged particles within a radius $R=|x|$, and we have abused notation to define $\rho_p^R(r)$ in the obvious way.

Since we have fixed $N_p$ and $N_e$ to satisfy the strict inequality version of (\ref{5/3bounds}), we can conclude that there is some $\gamma$ such that 
\begin{equation}
    q^2(N_p-N_e)-Gm_p(m_pN_p+M_eN_e)<\gamma<0
\end{equation}
Then we have 
\begin{equation}
    \lambda_p^R< \frac{5}{3}k_p(\rho_p^R)^{2/3}(x) +\frac{\gamma}{|x|}.
\end{equation}
Again, we can find $\overline{x}$ so that $\rho_p^R(\overline{x})<k_1R^{-3}$. Therefore, we can take $R$ large enough so we can find $\overline{x}$ showing that $\lambda_p^R<\hat{l}<0$ for some $\hat{l}$ and large $R$. 

Then where $\rho_p^R>0$ we have 
\begin{equation}
    \frac{q^2(N_p(|x|)-N_e)}{|x|}-\frac{Gm_p(m_pN_p(|x|)+m_eN_e)}{|x|}<\lambda_p^R<\hat{l}<0
\end{equation}
Clearly, this cannot be satisfied for large $|x|$. So for very large $R$, we must have both densities equal to zero. The same could be concluded in the case that $\rho_p^R$ becomes zero first. 

Now the standard argument found in \cite{AB1} shows that for any admissible pair $(N_e, N_p)$, their exists at least one minimizer of $E$ in $W$ with compact support.  

\subsection{The Special Relativistic Case}
Just as in the first proof of the existence of minimizers of $E$, this proof can be easily modified to accommodate the relativistic kinetic energy (\ref{relenergy}). Again, the proof only directly uses the kinetic energy in a few places: we address those here.
\begin{enumerate}
    \item In lemma \ref{compactminimizer}, we used that $x^{5/3}$ is convex in proving lower semicontinuity. $g(z)$ is also convex, so this step follows exactly as before. 
    \item In lemma \ref{43bounds}, we used that for any $\hat{c}$, $k_es^{5/3}\geq (\hat{c}+2\epsilon)s^{4/3}$ for $s>s_\epsilon$. Just as in the modifications to the first proof, this must be changed. In the proof, $\hat{c}=CG2^{1/3}m_p^{4/3}(m_eN_e+m_pN_p)^{2/3}$, were $C$ is the best constant in Proposition 6 of \cite{AB1}. For this step of the proof to go through, we need $$\hat{c}< \frac{\pi m_f c^5}{3h^3}\lim_{s\rightarrow \infty}\frac{(g(\alpha_fs^{1/3}))}{s^{4/3}}$$ which is true if \begin{equation}
        (m_eN_e+m_pN_p)^{2/3}<\frac{\pi 2^{2/3}hc\left(\frac{3}{8\pi}\right)^{4/3}}{GKm_p^{4/3}}
    \end{equation}
    This is the exact same bound as in the statement of theorem \ref{specialcase}.
    \item There are a few times in the remainder of the proof where we bound $\frac{d}{d\rho}(\rho^{5/3})$ by bounding $\rho$ by small numbers. These steps are also the same as for small $z$, $g(z)$ behaves like $z^{5/3}$, so similar estimates hold.  
\end{enumerate}
Then we have the theorem 
\begin{theorem}
If 
\begin{equation}
  \frac{1-\frac{Gm_p^2}{q^2}}{1+\frac{Gm_em_p}{q^2}}<\frac{N_e}{N_p}<\frac{1+\frac{Gm_em_p}{q^2}}{1-\frac{Gm_q^2}{q^2}},  
\end{equation}
 there is a pair $(\rho_e, \rho_p)$ which minimizes $E^s$ in $W$. Moreover, both $\rho_e$ and $\rho_p$ have compact support. 
\end{theorem}

\section{Minimizers in the Nonrelativistic Case}\label{ELsection}
Let us now turn to the study of the equations (\ref{nel1}) and (\ref{nel2}). Based on the analysis of the ``6/3'' model presented in \cite{HK}, we expect there will be two problems we will need to study. The first, following the terminology introduced in \cite{HK}, we call the bulk problem. For this problem, both densities are positive, and we have the system given by (\ref{nel1}) and (\ref{nel2}). The second problem, which we call the atmospheric problem, occurs when one of the densities is 0. Depending on which density is zero, we ignore one of (\ref{nel1}) or (\ref{nel2}), and use $\rho_f=0$ appropriately with the equation we focus on.

\subsection{The Bulk Problem}\label{bulkprob}
With spherical symmetry and the substitution $u_f^{2/3}=\rho_f$, we can rewrite (\ref{nel1}) and (\ref{nel2}) as 
\begin{equation}\label{sphereEL1}
u_p''(r)=-\frac{2}{r}u_p'+\frac{12\pi}{5k_p}\left(q^2-Gm_p^2\right)u_p^{3/2}(r)-\frac{12\pi}{5k_p}\left(q^2+Gm_pm_e\right)u_e^{3/2}(r)    
\end{equation}
and
\begin{equation}\label{sphereEL2}
u_e''(r)=-\frac{2}{r}u_e'-\frac{12\pi}{5k_p}\left(q^2+Gm_pm_e\right)u_p^{3/2}(r)+\frac{12\pi}{5k_p}\left(q^2-Gm^2_e\right)u_e^{3/2}(r).  
\end{equation}
We want to solve this system with initial conditions
\begin{equation}\label{BulkIC}
    u_p(0)=\alpha>0, u_e(0)=\beta>0, \text{ and } u'_e(0)=u_p'(0)=0;
\end{equation} 
the conditions on the derivatives necessary so we do not get a singularity at the origin. We will first prove the well-posedness of this initial value problem.
\begin{lemma}
While both $u_e$ and $u_p$ are greater than zero, the system given by (\ref{sphereEL1}), (\ref{sphereEL2}), and (\ref{BulkIC}) is locally well-posed.
\end{lemma}
\begin{proof}
Here we follow the approach of \cite{Serrin1}. Fix $(\alpha, \beta)$. We will show that for any positive $\alpha$ and $\beta$, there is a $\delta>0$ such that $u_e$ and $u_p$ are at least $C^2$ on $[0,\delta]$, and depend continuously on the initial data. Let $$X=\{(f,g)\in C([0,\delta])\times C([0,\delta])|\max (||f-\alpha||_\infty, ||g-\beta||_\infty )<\epsilon\},$$ where $\epsilon$ is chosen small enough so that $f,g>0$. Then consider $T:X\rightarrow X$ defined as 
\begin{align}\label{exteriorFP}
    T(f,g)&=(\hat{f}(r), \hat{g}(r))=\left(\alpha-\int_0^r\int_0^t \frac{Ff^{3/2}(s)-Eg^{3/2}(s)}{t^2}s^2dsdt, \beta-\int_0^r\int_0^t \frac{Bg^{3/2}(s)-Af^{3/2}(s)}{t^2}s^2dsdt\right)
\end{align}
For $F=\frac{12\pi}{5k_p}\left(q^2-Gm_p^2\right)$, $E=\frac{12\pi}{5k_p}\left(q^2+Gm_pm_e\right)$, $B=\frac{12\pi}{5k_e}\left(q^2-Gm^2_e\right)$, and $A=\frac{12\pi}{5k_e}\left(q^2+Gm_pm_e\right)$. It is easy to see that any such $\hat{f}$ and $\hat{g}$ are $C^2$, and that $T$ maps $X$ into itself if $\delta$ is small enough. For we can say 
\begin{align}
    |\alpha-\hat{f}(r)|\leq [F(\alpha+\epsilon)^{3/2}+E(\beta+\epsilon)^{3/2}]\int_0^r\int_0^t\frac{s^2}{t^2}dsdt=\frac{F(\alpha+\epsilon)^{3/2}+E(\beta+\epsilon)^{3/2}}{6}r^2
\end{align}
which can be made smaller than $\epsilon$ for $\delta$ small enough. This is a crude estimate, but for fixed $\alpha, \beta$, and $\epsilon$, it gives a lower bound on the length of the interval. Note that, for fixed $\epsilon$, this lower bound varies continuously in $\beta$ and $\alpha$, so that for small changes in $(\alpha, \beta)$, we may assume some minimal interval length. All of this is of course also true for $\hat{g}$. So $T$ maps $X$ to itself. It is also straightforward to check that a fixed point of (\ref{exteriorFP}) is a solution to (\ref{sphereEL1}) and (\ref{sphereEL2}). Then we next must show that $T$ is a contraction on $X$. This can be seen from 
\begin{align*}
    |\hat{f}_1(r)-\hat{f}_2(r)|&=\left|\int_0^r\int_0^t\frac{(Ff_1^{3/2}(s)-Eg_1^{3/2}(s))-(Ff_2^{3/2}(s)-Eg_2^{3/2}(s))}{t^2}s^2dsdt\right|\\
    &\leq\int_0^r\int_0^t\frac{|(Ff_1^{3/2}(s)-Eg_1^{3/2}(s))-(Ff_2^{3/2}(s)-Eg_2^{3/2}(s))|}{t^2}s^2dsdt\\
    &\leq ||(f_1,g_1)-(f_2,g_2)||_X(F+E)(\alpha+\beta+2\epsilon)^{1/2}\frac{r^2}{4}
\end{align*}
which is a contraction for $\delta$ small enough. We can say the same thing about $g$. Then we map apply the Banach Fixed Point Theorem to conclude that we have a unique classical solution to (\ref{sphereEL1}) and (\ref{sphereEL2}).

To prove that the solutions depend continuously on the initial data, we need to show that given some $(\alpha, \beta)$ and $\eta>0$, there is a $\delta$ and $r_0>0$ such that $||(\alpha, \beta)-(\hat{\alpha}, \hat{\beta})||_\infty<\delta$ implies $||(f,g)-(\hat{f}, \hat{g})||_\infty<\eta$ on $[0,r_0]$, where $(f, g)$ and $(\hat{f}, \hat{g})$ are the unique solutions corresponding to the initial data $(\alpha, \beta)$ and $(\hat{\alpha}, \hat{\beta})$, respectively. For any $(\alpha,\beta)$ both greater than zero, by the above statements, we know that if $\delta$ is small enough, we can assume some interval $[0, r_1]$ on which they both exist, and $\max (||f-\alpha||_\infty, ||g-\beta||_\infty, ||\hat{f}-\hat{\alpha}||_\infty, ||\hat{g}-\hat{\beta}||_\infty)<\epsilon$. Then we may say that $|f(r)-\hat{f}(r)|<2\epsilon+\delta$, and the same for $g$. This gives us that 
\begin{equation}
    |f(r)-\hat{f}(r)|<\delta+\frac{r^2}{4}(\alpha+\beta+\epsilon+\delta)^{1/2}[F+E](2\epsilon+\delta),
\end{equation}
which we can make smaller than any $\eta$ by choosing $r_0$ and $\delta$ small enough. The same is true for $g$, so we find that the function values, at least in some nonzero length interval, depend continuously on the initial data.

Of course, this only holds for a small interval around zero. But once we consider intervals $[R_0, R_1]$ for $R_0>0$, we can use more standard theorems as now the $\frac{2}{r}$ is bounded below and the left hand side of the ODE is Lipschitz. Thus we may conclude that in the bulk problem, there is a unique solution to the ODE system, and the densities at a particular distance from the origin $u_e(r)$ and $u_p(r)$ depend on the initial values in a continuous way. 
\end{proof}
Thus we have local well-posedness, so we now turn to global existence. For the solution to a given $(\alpha,\beta)$, define $\psi(r)=-Au^{3/2}_p(r)+Bu_e^{3/2}(r)$ and $\phi(r)=-Eu_e^{3/2}(r)+Fu^{3/2}_p(r)$ so that (\ref{sphereEL1}) and (\ref{sphereEL2}) can be rewritten as 
\begin{equation}\label{reduced1}
    u''_e(r)+\frac{2}{r}u'_e(r)=\psi(r)
\end{equation}
and
\begin{equation}\label{reduced2}
    u''_p(r)+\frac{2}{r}u'_p(r)=\phi(r)
\end{equation}
respectively.
\begin{lemma}\label{ICconditions}
Only solutions such that $\phi(0), \psi(0)<0$ can have nonincreasing densities. 
\end{lemma}
\begin{proof}
In this proof, and in many of the subsequent proofs, we use a sort of crude maximum principle. Suppose that $\phi(0)>0$ and $\psi(0)<0$. Since $\lim_{r\rightarrow0}\frac{u'_e(r)}{r}=\frac{2}{3}\psi(0)$, and $\lim_{r\rightarrow0}\frac{u'_p(r)}{r}=\frac{2}{3}\phi(0)$ we can say that $u''_p(0)>0$ and $u''_e(0)<0$. Then by continuity, $u_p$ is increasing and $u_e$ is decreasing on a short interval. But then on the same interval, $\phi$ is increasing and $\psi$ is decreasing. By our local existence result, let us continue this solution until either $u_e'(r)=0$ or $u_p'(r)=0$; the conclusion will be the same if neither happens. In the former case, we can see from (\ref{reduced1}) that $u''_e(r)<0$, so $u_e'$ will continue decreasing, while (\ref{reduced2}) tells us that in the latter case $u_p$ will continue increasing. So such starting values give solutions for which $u_e$ is nonincreasing and $u_p$ is nondecreasing, at least while both are positive. We have a similar outcome if $\phi(0)<0$ and $\psi(0)>0$.

To get $\phi(0)>0$ and $\psi(0)>0$, we must have 
\begin{equation}\label{modbound1}
\psi(0)=-A\alpha^{3/2}+B\beta^{3/2}>0
\end{equation}
and 
\begin{equation}\label{modbound2}
\phi(0)=-E\beta^{3/2}+F\alpha^{3/2}>0
\end{equation}
Combining these inequalities leads to the inequality $BF>EA$, but these are all physical constants, and it can be checked that this inequality is not satisfied for our problem. Therefore, we may discount this case.

There are some boundary cases we also need to consider. The case $\phi(0)=\psi(0)=0$ and the cases $\psi(0)=0$, $\phi(0)>0$ and $\psi(0)>0$, $\phi(0)=0$ all cannot happen for the same reason that $\psi(0)>0$, $\phi(0)>0$ cannot happen. The case $\phi(0)<0$, $\psi(0)=0$ becomes the same as the case $\phi(0)<0$, $\psi(0)>0$, while the case $\phi(0)=0$, $\psi(0)<0$ becomes the same as the case $\phi(0)>0$, $\psi(0)<0$.
\end{proof}
Although we can find solutions to (\ref{sphereEL1}) and (\ref{sphereEL2}) with initial conditions such that either $\phi(0)$ or $\psi(0)$ is not negative (the situation considered in the previous lemma), we will show below in the subsection on the atmospheric problem that such solutions cannot be integrable. 

We are left to consider cases which satisfy 
\begin{equation}\label{modbound3}
    \phi(0)=-E\beta^{3/2}+F\alpha^{3/2}<0
\end{equation}
 and 
 \begin{equation}\label{modbound4}
     \psi(0)=-A\alpha^{3/2}+B\beta^{3/2}<0
 \end{equation}
 For this to be satisfied, we need to have $EA>BF$, which we know is true. Similar to the ratio of the total number electrons to protons, (\ref{modbound3}) and (\ref{modbound4}) severely restrict the cases under consideration, as we must have both $\alpha<(E/F)^{2/3}\beta$ and $\alpha>(B/A)^{2/3}\beta$. In fact, if we write these out, we see they are essentially (\ref{bounds}):
\begin{equation}\label{densityratios}
    (q^2)^{2/3}\left[\frac{1-\frac{Gm_q^2}{q^2}}{1+\frac{Gm_em_p}{q^2}}\right]^{2/3}<\frac{\alpha}{\beta}<(q^2)^{2/3}\left[\frac{1+\frac{Gm_em_p}{q^2}}{1-\frac{Gm_p^2}{q^2}}\right]^{2/3}
\end{equation}

In describing the global structure of the solutions to the bulk problem, we may ask whether there is a ``special solution'', similar to the ``no atmosphere'' solution presented in \cite{HK}: that is, a solution to this system such that the densities are proportional to each other. In fact, we can construct such solutions in the general case where we have the system 
\begin{equation}\label{nonlinear1}
\Delta (u^{(d-3)/3})=-Av+Bu
\end{equation}
and 
\begin{equation}\label{nonlinear2}
\Delta (v^{(d-3)/3})=-Eu+Fv
\end{equation}
for $d>3$ (still using the physical constants).
\begin{lemma}\label{specialsolution}
For $d>3$, the system given by (\ref{nonlinear1}), (\ref{nonlinear2}), and 
\begin{equation}
    u(0)=\alpha>0, v(0)=\beta>0, \text{ and } u'(0)=v'(0)=0;
\end{equation} 
has a solution such that $u$ is proportional to $v$. If $d>\frac{18}{5}$, there is such a solution with compact support. If $3<d<6$, there is only one such solution. 
\end{lemma}
\begin{proof}
Note that this is the system of PDEs we would get from the Euler-Lagrange equations of our energy functional if our kinetic energy term was proportional to $\int_{\mathbb{R}^3}u^{d/3}$. Let us assume $u=k_dv$. Plugging this into (\ref{nonlinear1}) and (\ref{nonlinear2}), we get a system which allows us to solve for $k$. We find that our desired $k_d$ is a solution to 
\begin{equation}\label{specialEquation}
H_d(k)=k^{d/3}F+Ak-Ek^{(d-3)/3}-B=0
\end{equation}
For $3<d<6$, $H_d$ has a single root, which is easy to see as $H_d(0)<0$ while $\lim_{k\rightarrow\infty}H_d(k)=\infty$ and $H_d''>0$. 

Now that we have a $k_d$ for each $d$, we need to solve $\Delta v^{(d-3)/3}=(-E+k_dF)v$, where we will assume spherical symmetry. As above, let us set $v^{(d-3)/3}=\nu$. Then in spherical coordinates, this equation becomes 
\begin{equation}\label{LaneEmden}
\nu''(r)+\frac{2}{r}\nu(r)=(-E+k_dF)\nu^{3/(d-3)}
\end{equation}
How we can treat this equation depends on the sign of $-E+k_dF$. We can determine this by inserting $E/F$ into $H_d(k)$. We get $$\left(\frac{E}{F}\right)^{d/3}F+A\frac{E}{F}-E\left(\frac{E}{F}\right)^{(d-3)/3}-B=A\frac{E}{F}-B>0.$$ Since $H_d$ is convex and has a single root, we can then conclude that $E/F>k_d$ which implies that $-E+k_dF<0$. But this means that (\ref{LaneEmden}) can be transformed into the Lane-Emden equation of index $3/(d-3)$. For the range $3<d\leq6$, this covers indices of $[1,\infty)$. It is well known that the Lane-Emden equation has solutions, although only for indices less than or equal to 5 are they integrable. 
\end{proof}

The Lane-Emden equation has analytic solutions when the index is 1 or 5 (also 0, but this is not in our range), which correspond to $d=6$ or $d=\frac{18}{5}$. The case under consideration here is for $d=5$, while \cite{HK} presents the linear case $d=6$. As has been mentioned, \cite{HK} discusses $d=6$ fully. The $d=\frac{18}{5}$ case has found some applications in General Relativity, but we will not consider it further. Our $d=5$ case corresponds to an index of $3/2$, solutions of which are known to have compact support.

As the physical interpretation of the Lane-Emden equation is the density of a self-gravitating polytropic fluid, we can then interpret our solutions which are close to the special solutions as perturbations of this. As a follow up to the zero-G model discussed in section \ref{themodel}, one can easily see what happens to the polytropic solutions when $G\rightarrow0$. Setting $G=0$ means that $E=F$ and $B=A$. As $H_d(k)$ has a single positive root, we must have $k_d=1$. So as $G\rightarrow 0$, $k_d\rightarrow1$. Therefore, when we look for the special solutions of the zero-G model, we get $u=v$, and the electric energy also vanishes. This leaves only the kinetic energy. Obviously, this has no nonzero minimizer, so we get the same result as above.

Since we would expect the proton density to be almost equal to the electron density, we may seek solutions of the form $u=v$ as approximations. These will not be solutions of (\ref{nonlinear1}) and (\ref{nonlinear2}); even in the case when $G=0$, this produces $u=v=\gamma$ for $\gamma$ any constant, which clearly is not integrable. If we do want densities of this form, we instead minimize the energy functional
\begin{equation}
    E_n(u)=(k_e+k_p)\int_{\mathbb{R}^3}u^{d/3}(s)d^3s-G(m_e+m_p)^2\int_{\mathbb{R}^3}\int_{\mathbb{R}^3}\frac{u(s)u(t)}{|s-t|}d^3sd^3t
\end{equation}
The Euler-Lagrange equation is 
\begin{equation}
    -\frac{d}{3}(k_e+k_p)u^{(d-3)/3}(s)-G(m_e+m_p)^2\int_{\mathbb{R}^3}\frac{u(t)}{|s-t|}d^3t+\lambda=0
\end{equation}
So applying $-4\pi\Delta$ as above, and letting $u^{(d-3)/3}=\nu$, we get 
\begin{equation}
    \frac{d}{3}(k_e+k_p)\Delta \nu(s)=-4\pi G(m_e+m_p)^2\nu^{3/(d-3)}(s)
\end{equation}
But if we assume spherical symmetry, this again gives us a (scaled) Lane-Emden equation of index $(d-3)/3$. This scaling will be different from that for the special solution of (\ref{nonlinear1}) and (\ref{nonlinear2}), so the densities will be different although have the same shape. This locally neutral approximation is precisely what has been done for decades, and is exactly the assumption that we are removing in our model. 

Returning to our main problem, from our assumption that $\phi(0)<0$ and $\psi(0)<0$, we see that both $u_e$ and $u_p$ in a local neighborhood of 0 are nonincreasing. Applying our local existence result, we continue solving outwards until at least one of $u_p$ or $u_e$ vanishes or has positive slope. Before we consider what cases this gives us, we must show that these are the only possibilities. 
\begin{lemma}
There is no solution such that $\phi(0)<0$, $\psi(0)<0$, and both densities are nonincreasing with unbounded support. 
\end{lemma}
\begin{proof}
Let $\lim u_p=u_p^\infty$ and $\lim u_e=u_e^\infty$. From (\ref{reduced1}) and (\ref{reduced2}) we see that both $\psi$ and $\phi$ must have a limit of zero. This is only possible if both $u_p^\infty$ and $u_e^\infty$ are zero. To eliminate this possibility, we take advantage of the special solutions discussed above. Suppose that $(v_1, u_1)$ is a solution of this form. Then we can find a special solution $(v_2, u_2)$ such that either $u_1(0)=u_2(0)$ and $v_1(0)\neq v_2(0)$ or $u_1(0)\neq u_2(0)$ and $v_1(0)=v_2(0)$; uniqueness prevents both equalities. 

There are then four cases, but all can be treated in the same way. Consider the case that $u_1(0)=u_2(0)$ and $v_1(0)>v_2(0)$. We compute 
\begin{equation}
    (u_1-u_2)''(0)=F(u_1^{3/2}(0)-u_2^{3/2}(0))-E(v_1^{3/2}(0)-v_2^{3/2}(0))<0
\end{equation}
so that $(u_1-u_2)$ is decreasing for a short interval around the origin, implying that on that interval, $u_2>u_1$. We also compute
\begin{equation}
    (v_1-v_2)''(0)=B(v_1^{3/2}(0)-v_2^{3/2}(0))-A(u_1^{3/2}(0)-u_2^{3/2}(0))>0
\end{equation}
So $v_1$ is initially greater than $v_2$, and this difference is on some interval increasing. Now we apply the existence result to continue these solutions until either $u_1'=u_2'$ or $v_1'=v_2'$; it is clear that $u_1<u_2$ and $v_1>v_2$ at any distance less than or equal to this. If this never happens, we get a contradiction since $u_1<u_2$, but $u_2$ vanishes at some finite distance. If $u_1'=u_2'$, (\ref{sphereEL2}) gives us that $(u_1-u_2)''<0$, and if $v_1'=v_2'$, (\ref{sphereEL1}) gives us that $(v_1-v_2)''>0$. So $u_1\leq u_2$, which again gives a contradiction since $u_2$ has compact support. The other three cases are treated similarly, and allow us to conclude that there is no solution to the bulk problem with unbounded support.  
\end{proof}

This same technique can be used to show that the special solutions are the only solutions such that the densities vanish at the same point. 
\begin{lemma}\label{uniquespecial}
For $d=5$, the special solution given by lemma \ref{specialsolution} is the only solution such that the densities have the same support. 
\end{lemma}
\begin{proof}
Suppose $(u_1, v_1)$ is another such solution. We can again find $(u_2, v_2)$ such that either $u_1(0)=u_2(0)$ and $v_1(0)\neq v_2(0)$ or $u_1(0)\neq u_2(0)$ and $v_1(0)=v_2(0)$. Assuming the case that $u_1(0)=u_2(0)$ and $v_1(0)>v_2(0)$, we can again conclude that $u_1\leq u_2$ and $v_1\geq v_2$. In fact, since we can say that $u_1'\leq u_2'$ and at some point $u_1'<u_2'$, we can conclude that $u_1<u_2$. With the same reasoning $v_1>v_2$. But since $v_2$ and $u_2$ vanish at the same point, we have a contradiction. As the other cases are similar, we conclude that the special solutions are the unique solutions such that the densities vanish at the same point. 
\end{proof}

Then we have returned to our stopping criteria. In the simplest case, one of the densities vanishes first, and we find ourselves in the exterior problem, to be elaborated on below. The case when both vanish at the same distance must be a special solution, as shown above. 
\begin{lemma}\label{increasedecrease}
If $u_f'(R_0)>0$ at $R_0>0$, then for $R>R_0$, $u_f'\geq 0$ while $u_{\sim f}'\leq 0$. 
\end{lemma}
\begin{proof}
Let us assume that $u_p'$ becomes positive before $u_e'$. By continuity, $u_p'$ must first be zero, so let us consider what can happen at $R_0$, the last point at which $u_p'=0$ before it becomes positive. We are going to examine all the possibilities for different values of $u_p$, $u_e$, $u_p'$, and $u_e'$ at $R_0$ and show that all have the same behavior.

First, consider if $u_e'(R_0)=u_p'(R_0)=0$. We must always have one of $\psi$ or $\phi$ be strictly negative, and since we are assuming $u_p'$ becomes positive first, this must be $\psi$. Therefore $u_e''(R_0)<0$, and on some interval $[R_0, R_0+\delta]$, $u_e'$ is negative. If $\phi(R_0)>0$, then $u_p'$ is positive on an interval also. If $\phi(R_0)=0$ then $u_p''(R_0)=0$, and we use $\psi(R_0)<0$ to conclude that $u''(R_0)<0$. But this means that $\phi''(R_0)>0$. Since $\phi'(R_0)=0$, we again conclude that $u_p$ is increasing while $u_e$ is decreasing on some interval $[R_0, R_0+\delta]$.

Now we may consider when $u_p'(R_0)=0$ and $u'_e(R_0)<0$. If $\phi(R_0)>0$ then $u_p''(R_0)>0$ and again we find that on some interval $[R_0, R_0+\delta]$ $u'_e$ is decreasing while $u'_p$ is increasing. If $\phi(R_0)=0$, then $u_p''=0$. We must then account for the behavior of $u_e$. Since $u_e'(R_0)<0$, then $\phi'(R_0)>0$. This implies that $u_p'\geq0$ on an interval, and we come out with the same result.  

So we can conclude that their is some interval $[R_0, R_0+\delta]$ such that $u_p$ is increasing and $u_e$ is decreasing. It remains to conclude that, as long as (\ref{sphereEL1}) and (\ref{sphereEL2}) are valid (when both densities are positive), $u_p$ is nondecreasing and $u_e$ is nonincreasing.  Note that in all cases above, on this same short interval, $\phi\geq 0$ and $\psi\leq 0$. $u_p'$ cannot become negative without first being zero and $u_e'$ cannot become positive without first being zero. Suppose we continue our solutions until the first time one of these happens, again suppose it is $u_p'(R_1)=0$. Then $u_p''(R_1)$ will have the same sign as $\phi$. But $\phi$ cannot decrease since $u_e$ is decreasing and $u_p$ is increasing on the interval $[R_0, R_1]$. Therefore $\phi(R_1)\geq 0$, and we can conclude that $u_p'$ cannot become negative. The same argument is true for $u_e'$ not becoming positive. Therefore, for $R>R_0$, $u_e'(R)\leq 0$ and $u_p'(R)\geq 0$. Obviously, the same result holds if $u_e'$ first becomes positive.  
\end{proof}

To conclude our classification of the bulk problem, we must then determine what happens when either $u_p'$ or $u_e'$ becomes positive. 
\begin{lemma}\label{boundeddensity}
If either density is at some point increasing, the other density will vanish, and the increasing density will remain bounded on the vanishing density's support.  
\end{lemma}
\begin{proof}
We have already seen one density becomes nonincreasing while the other is nondecreasing. Only three things can happen from here. For simplicity, assume that $u_p'>0$. In the first situation, $u_p$ stays bounded while $u_e$ vanishes with compact support; in the second situation, $u_p$ develops a singularity while $u_e>0$; and in the third situation, both $u_p$ and $u_e$ have support on the entire line. We will here prove that the second and third options cannot happen. To show that the second situation cannot happen, we multiply (\ref{sphereEL1}) by $A/F$ and add this to (\ref{sphereEL2}). This gives us 
\begin{equation}
    (u_e''+\frac{A}{F}u_p'')+\frac{2}{r}(u_e'+\frac{A}{F}u_p')=(B-\frac{A}{F}E)u_e^{3/2}
\end{equation}
If $u_p$ developed a singularity at some $R_0$, then the left hand side would become unbounded as $r\rightarrow R_0$, while the right hand side would remain bounded, a contradiction. To see that the third option cannot happen, first note that this case implies that $\psi$ has a limit of zero. We can therefore conclude that $u_p$ must be bounded and has some finite limit as $r\rightarrow\infty$. This implies that $\lim_{r\rightarrow\infty}\phi(r)=(\frac{B}{A}F-E)\lim_{r\rightarrow\infty}u_e^{3/2}(r)$. Since $(\frac{B}{A}F-E)<0$ but $\phi\geq 0$, we must have $\lim_{r\rightarrow\infty}u_e^{3/2}(r)=0$. Of course, this implies $\lim_{r\rightarrow\infty}u_p^{3/2}(r)=0$, a contradiction. Then we have eliminated the second and third scenarios, so only the first scenario is possible. 
\end{proof}
Let us summarize our results for the bulk problem:
\begin{enumerate}
    \item We cannot have both densities increasing. This follows from lemma \ref{ICconditions} and lemma \ref{increasedecrease}
    \item If a density is increasing, it is bounded on the decreasing density's support, which must be compact. This is lemma \ref{boundeddensity} 
    \item The only solutions for which the densities have the same support are the special solutions. This is lemma \ref{uniquespecial}
\end{enumerate}
So given any initial values for the bulk problem, if we do not have a special solution, we see that one of the densities will vanish at some finite radius $R_0$. This transfers directly to the atmospheric problem.
\subsection{The Atmospheric Problem}\label{atmosprob}
For the atmospheric problem, we are trying to solve the initial value problem
\begin{align}\label{atmosphere}
    \Delta u &=Du^{3/2} \text{ on } (R_0, R_1)\\
    u(R_0)&=\gamma>0
\end{align}
Where $D>0$, $R_0>0$, and $R_1$ is either where $u$ vanishes, or $\infty$. With radial symmetry, it is a standard exercise to show that this problem is locally well-posed. We therefore focus only on the global structure of the solutions.

There are a few simple observations which can be made. First, if $u'$ becomes zero on the support of $u$, then $u$ will always be nondecreasing after that point. This is easy to see as if $u''=Du^{3/2}>0$, the first derivative must be increasing. This tells us that we cannot have any oscillatory behavior and that any solution which ever hits zero must have always been decreasing. When we combine the bulk and atmospheric problem, this will also prove the statement that we can ignore any solution of the bulk problem for which one of the densities is nondecreasing. 

Second, let us consider what the limiting values of $u$ may be. In particular, we want to conclude that there are no solutions such that $u(r)\rightarrow\alpha>0$ as $r\rightarrow \infty$ (since there is no oscillatory behavior, we know $u$ is monotonic, so such a limit exists). For suppose that such a solution exists, then $u$ may approach $\alpha$ from above or below. In either case, we must have $u'(r)\rightarrow0$, so for $r$ larger than some $R$, we can conclude that $u''(r)>\alpha-\epsilon$>0 contradicting that $u'(r)\rightarrow 0$. So if $u$ remains positive, it can only have limits of $\infty$ or 0.

Further, there may be only one strictly positive solution which approaches zero at infinity. For suppose that $u_1$ and $u_2$ were two such solutions, and let $u_1'(R_0)>u_2'(R_0)$. Then $u_1>u_2$ on $(R_0, R]$ for any finite $R$ by the maximum principle, and the difference is largest at $R$. But by assumption, if we take $R$ large enough we must have both $u_1(r)$ and $u_2(r)$ in $(0, \epsilon)$ for any $\epsilon>0$. Thus we can force $u_1'(R_0)$ to be arbitrarily close to $u_2'(R_0)$, so we must have equality and $u_1=u_2$. Call this possible solution $u_\infty$. This shows that for fixed $R_0$ and $u(R_0)$, if $u$ is a solution such that $u'(R_0)>u'_\infty(R_0)$, $u$ must be unbounded.  

Finally, concerning solutions which have compact support,  we can say that if we have two solutions $u_1$ and $u_2$ satisfying $u_1(R_0)=u_2(R_0)$, the one with the larger slope at $R_0$ must hit zero at a larger $r$. This follows from a maximum principle as the solution with a larger initial slope must have a larger function value than the other solution, at least as long as one of the functions is positive. This tells us that the initial slopes have a linear ordering: larger initial slopes lead to larger radii when the density becomes zero, if it becomes zero at all. However, nothing above tells us that any solution ever hits zero, so that is what we turn to now. 

If we pose this question as the boundary value problem, 
\begin{align}\label{exteriorPDE}
    \Delta u &=f(u) \text{ on }\Omega\\
    u&=g\text{ on }\partial\Omega
\end{align}
where $\Omega$ is the annular domain with inner radius $R_0$ and outer radius $R$, $g=a>0$ on the inner radius and $g=0$ on the outer radius, it is a standard exercise that a classical unique solution exists. Since the boundary values are spherically symmetric, by uniqueness, so is the solution. Therefore, the solution found here will be the same as the unique solution to the ODE problem with the corresponding $u'(R_0)$ value.

From the above arguments, we may conclude that given a fixed $a>0$ and $R_0>0$, for every $R>R_0$, we can find a unique strictly decreasing solution $u_R$, and that $u_{R_1}(r)\geq u_{R_2}(r)$ if $R_1>R_2$, with the inequality strict for $R_0<r\leq R_2$. Now, we want to prove that there is a solution (we have already stated it must be unique) with unbounded support which has $\lim_{r\rightarrow\infty}u(r)\rightarrow0$. 
\begin{theorem}\label{limitingcase}
For any given $R_0$ and $a>0$, there is a unique spherically symmetric solution to $\Delta u=Du^{3/2}$ which is both strictly positive and satisfies $\lim_{r\rightarrow\infty}u(r)=0$.
\end{theorem}
\begin{proof}
First, since each $u_R$ is strictly decreasing and $u_{R_1}(r)\geq u_{R_2}(r)$ if $R_1>R_2$, $\lim_{R\rightarrow\infty} u_R(r)$ exists for every $r$. Define this value to be $\hat{u}(r)$. The claim is that $\hat{u}(r)$ is the desired function. To show this, we are going to show that $u_R\rightarrow \hat{u}$ in $C^2{[R_0, R_1]}$ for any $R_0<R_1<\infty$. Then in particular we will have that $D^\alpha u_R(r)\rightarrow D^{\alpha}\hat{u}(r)$ for any $|\alpha|\leq 2$ and any $R_0<r<\infty$. So $\hat{u}$ will solve the ODE. 

To get this convergence, we use Rellich-Kondrachov to get a compact embedding of the set $\{u_R\}_{R>N}$ when we have restricted to $\Omega_N=(R_0, N)$. So fix $N$ and consider the set $\{u_R\}_{R>N}$. In particular, each of these functions is strictly greater than zero on $\Omega_N$. Then applying some elliptic regularity estimates, we conclude that each $u_R$ is in $H^3(\Omega_N)$ (meaning the $L^2$ Sobolev space) and further satisfies 
\begin{equation}\label{ellipticregularity}
||u_R||_{3, \Omega_N}\leq C(||u_R^{3/2}||_{1,\Omega_N}+||u_R||_{0,\Omega_N}+||u_R||_{3, \partial\Omega_N})    
\end{equation}
Thinking of $\Omega_N$ as a set in $\mathbb{R}^3$, the last term is bounded as $u_R$ is less or equal to $a$ on the boundary. The middle term is also bounded by $a$ multiplied by the size of $\Omega_N$. To bound the first term, we need to have some control over $u_R'$, and to do this, we apply the same estimate one derivative lower: $$||u_R||_{1, \Omega_N}\leq||u_R||_{2, \Omega_N}\leq C(||u_R^{3/2}||_{0,\Omega_N}+||u_R||_{0,\Omega_N}+||u||_{2, \partial\Omega_N})\leq K$$ for some fixed $K$, for all $R>N$. This allows us to get a uniform bound on (\ref{ellipticregularity}). So we not only have that $\{u_R\}_{R>N}$ is contained in $H^3(\Omega_N)$, we have that it is bounded as well. Then, treating these as functions on the real line, Rellich-Kondrachov gives us that this set embeds compactly into $C^{2}(\overline{\Omega})=C^{2}([R_0, N])$. In particular we get pointwise convergence as described above, so $\hat{u}$ is a solution to the ODE. 
\end{proof}
$\hat{u}$ is then a sort of boundary case to solutions of (\ref{atmosphere}): solutions with $u'(R_0)<\hat{u}'(R_0)$ have compact support while solutions with $u'(R_0)>\hat{u}'(R_0)$ are unbounded. Although not interesting from a physical perspective, we can say more about these unbounded solutions.
\begin{lemma}
Any solution greater than $\hat{u}$ develops a singularity. 
\end{lemma}
\begin{proof}
The key here is what we have already proven about these solutions: they cannot remain bounded as $r\rightarrow\infty$. Then given any such solution of this type, we may assume that for large enough $r$, $u(r)>\gamma$ for some large $\gamma$. From the ODE, we may also conclude that $u'$ cannot be bounded. Fix a large $R$, and consider $f_R(r)=\frac{\omega}{(R+1-r)^4}$. If $\omega$ is chosen large enough, then on the interval $[R, R+1)$,
\begin{equation}\label{singularity}
    \Delta f_R(r)=\frac{20\omega}{(R+1-r)^6}+\frac{8\omega}{r(R+1-r)^5}\leq D\frac{\omega^{3/2}}{(R+1-r)^6}=Df_R^{3/2}(r)
\end{equation}
Note that $\omega$ is chosen independently of $R$. We have $\Delta(u(r)-f_R(r))=D(u^{3/2}{r}-f_R^{3/2}(r))>0$ on a small interval $[R, R+\delta]$ if $u(R)>f_R(R)$, and $u'(R)>f_R'(R)$. This implies that $u(r)>f_R(r)$ as long as (\ref{singularity}) holds. But since $u(r)$ and $u'(r)$ are unbounded, we may find an $R$ so that $u(R)>f_R(R)$ and $u'(R)>f'_R(R)$. Therefore $u(r)>f_R(r)$ on $[R, R+1)$, so $u(r)$ must develop a singularity.
\end{proof}
A final point we want to make is to show that $\hat{u}$ is integrable.
\begin{lemma}
$\hat{u}$ is integrable.
\end{lemma}
\begin{proof}
Recall that for each $R>0$ and each $a>0$ we get such a $\hat{u}$, so now we label them as $\hat{u}_{R,a}$. We will use that we have a special solution: $v(r)=\frac{C}{r^4}$ for $C=144/D^2$. Obviously, since the equation is not linear, we cannot scale this to get any other solutions. But we see that if $v(R)=a$, then $\hat{u}_{R,a}$ coincides with $v(r)$; this must be the case since we have already shown above that the solution to the ODE which has unbounded support and converges to zero is unique. For simplicity, fix $R$. Then we have two scenarios: either $a>v(R)$ or $a<v(R)$. 

In the latter case, suppose at some point $R_1$ $\hat{u}_{R,a}(R_1)=v(R_1)$.  By uniqueness of the solution converging to zero, both of these functions must after $R_1$ coincide. Therefore $\hat{u}_{R,a}\leq v$, and decays at least as fast as $r^{-4}$, which is integrable outside of a ball on $\mathbb{R}^3$. In the former case, we consider the function $\overline{v}(r)=\frac{\overline{c}}{r^4}$, where $\overline{c}$ is chosen so that $\overline{v}(R)>\hat{u}_{R, a}(R)=a$. Since $\overline{c}>144/D^2$, one can check that, treating the function as spherically symmetric on $\mathbb{R}^3$, $\Delta v\leq D\overline{v}^{3/2}$. Now suppose that there is some $R_1$ so that $\hat{u}_{R,a}(R_1)=\overline{v}(R_1)$. Then, for at least some small interval $[R_1, R_1+\delta]$, $\hat{u}_{R,a}> \overline{v}$ (we could not have equality since $\overline{v}$ does not satisfy the ODE). But then we have on that interval $\Delta (\hat{u}_{R,a}-\overline{v})\geq D(\hat{u}_{R,a}^{3/2}-\overline{v}^{3/2})> 0$. By the maximum principle, this implies that $\hat{u}_{R,a}(R_1+\delta)>\overline{v}(R_1+\delta)$, so we may extend the interval longer and the functions would always get further apart. But since they both converge to zero, this cannot happen. So we see that $\hat{u}_{R,a}<\overline{v}$, and therefore goes to zero at least as fast as $r^{-4}$ and is therefore integrable. So we have shown that in all scenarios $\hat{u}_{R,a}$ is an integrable function.
\end{proof}
\subsection{Combining the Problems}
Now that we have addressed the atmospheric problem and the bulk problem in isolation, we can bring them together to solve the system of equations. The method is clear. Given any pair of central densities which satisfy (\ref{densityratios}), we use the results of subsection \ref{bulkprob} to find the unique solution of the system continuing outward radially until one of the densities vanishes, say at $R_0$. The nonvanishing density will satisfy $u(R_0)=a>0$ and $u'(R_0)=b$. We then use these values as our initial values in the atmospheric problem and apply the results of subsection \ref{atmosprob}.

One issue which needs to be addressed is regularity. We have established that the solutions of both the bulk and atmospheric problems are $C^2$ on their domains. But the density which is nonzero in the atmosphere might have a jump in its second derivative at the interface of the two problems. To show that this is not the case, let us only consider those initial values such that $u^{3/2}_p=\rho_p$ and $u^{3/2}_e=\rho_e$ are integrable, recalling that in section \ref{themodel} we transformed the actual densities by the inverse of this. It is clear that these will be the densities which have compact support, as well as the limiting case discussed in theorem \ref{limitingcase}. Of course, for our minimization problem, all the other solutions are nonviable since one of the densities is unbounded. 

Before we proceed, we need the following extension of lemma 4.2 in \cite{GT}.

\begin{lemma}
Let $f$ be bounded, locally H\"older continuous with exponent $\alpha\leq 1$, and in $L^p$ for any $1\leq p<\infty$ on $\mathbb{R}^n$. Then if $w$ is the Newtonian potential of $f$,  $w\in C^2(\mathbb{R}^n)$.
\end{lemma}
\begin{proof}
Consider $f_R(y)=\chi_{B_R}f(y)$, and call $w_R$ its Newtonian potential. Then $f_R$ satisfies the hypotheses of lemma 4.2 in \cite{GT} in $B_R$, and so $w_R$ is $C^2(B_R)$. By that lemma, we have for any $x\in B_R$
\begin{equation}
    D_{ij}w_R(x)=\int_{B_R}D_{ij}\Gamma(x-y)(f_R(y)-f_R(x))dy-f_R(x)\int_{\partial B_R}D_i\Gamma(x-y)\nu_j(y)ds_y
\end{equation}
where $\Gamma$ is the fundamental solution of Laplace's equation. For $x$ such that $|x|<R-1$, rewrite this as 
\begin{align}
    D_{ij}w_R(x)&=\int_{B_1(x)}D_{ij}\Gamma(x-y)(f_R(y)-f_R(x))dy+\int_{B_R\backslash B_1(x)}D_{ij}\Gamma(x-y)(f_R(y)-f_R(x))dy\\
    &-f_R(x)\int_{\partial B_R}D_i\Gamma(x-y)\nu_j(y)ds_y \nonumber\\
    &=\int_{B_1(x)}D_{ij}\Gamma(x-y)(f_R(y)-f_R(x))dy+\int_{B_R\backslash B_1(x)}D_{ij}\Gamma(x-y)f_R(y)\\
    &-f_R(x)\int_{\partial B_1}D_i\Gamma(x-y)\nu_j(y)ds_y \nonumber
\end{align}
 For any selected $x$, we can define $D_{ij}w_R(x)$ like this for all $R>|x|+1$. Now define for $x\in \mathbb{R}^3$ 
 \begin{equation}
     u(x)=\int_{B_1(x)}D_{ij}\Gamma(x-y)(f(y)-f(x))dy+\int_{\mathbb{R}^n\backslash B_1(x)}D_{ij}\Gamma(x-y)f(y)dy-f(x)\int_{\partial B_1}D_i\Gamma(x-y)\nu_j(y)ds_y
 \end{equation}
 This is well defined because we can apply H\"older's inequality to the second term. Then for large $R$ we have 
 \begin{equation}
     |u(x)-D_{ij}w_R(x)|\leq \int_{ B^c_R(x)}\left|D_{ij}\Gamma(x-y)f(y)\right| dy\leq C \int_{ B^c_R(x)}\frac{f(y)}{|x-y|^n} dy
 \end{equation}
 which can be made as small as we like if we take $R$ large enough. Then we have that for any $x\in \mathbb{R}^n$, $D_{ij}w_R(x)\rightarrow u(x)$. Further, on any compact set in $\mathbb{R}^n$, this convergence is uniform. Since we also have that $w_R(x)\rightarrow w(x)$ uniformly on any compact set, we can conclude that $u=D_{ij}w$, and that $w$ is $C^2(\mathbb{R}^n)$.
\end{proof}
In the following lemma, we need to use that $\rho_e$ and $\rho_p$ are H\"older continuous. For a density with compact support, this will be true as long as its first derivative does not become unbounded as it approaches the boundary of its support. This is however ruled out by (\ref{sphereEL1}) and (\ref{sphereEL2}) since a large negative first derivative would lead to a positive second derivative. 
\begin{lemma}
When $\rho_e$ and $\rho_p$ are integrable, they are both $C^2$ on their supports. 
\end{lemma}
\begin{proof}
When $\rho_{f}$ have compact support, let us define new densities by 

\begin{equation}\label{el1def}
    \frac{5}{3}k_p\bar{\rho}_p^{2/3}(x):=-q^2\int_{\mathbb{R}^3}\frac{(\rho_p-\rho_e)(y)}{|x-y|}d^3y+Gm_p \int_{\mathbb{R}^3}\frac{(m_p\rho_p+m_e\rho_e)(y)}{|x-y|}d^3y+\lambda_p
\end{equation}
and
\begin{equation}\label{el2def}
    \frac{5}{3}k_e\bar{\rho}_e^{2/3}(x):=C q^2\int_{\mathbb{R}^3}\frac{(\rho_p-\rho_e)(y)}{|x-y|}d^3y+Gm_e \int_{\mathbb{R}^3}\frac{(m_p\rho_p+m_e\rho_e)(y)}{|x-y|}d^3y+\lambda_e
\end{equation}
where we take (\ref{el1def}) to hold on the support of $\rho_p$, $\bar{\rho}_p$ is defined to be zero outside of this support, and $\lambda_p$ is chosen to make $\bar{\rho}_p$ continuous. $\bar{\rho}_e$ is defined analogously. We can then apply lemma 4.2 in \cite{GT} to conclude that both $\bar{\rho}_e$ and $\bar{\rho}_p$ are $C^2$ on their supports. 

Now apply $-\Delta$ to (\ref{el1def}) and (\ref{el2def}). Suppose that $\rho_p$ is the density which does not vanish at $R_0$, but vanishes at $R_1$. Considering the atmospheric and bulk problem separately, we can conclude that
\begin{equation}\label{maximumprinciple}
\Delta(\bar{\rho}_{p}-\rho_{p})=0 \text{ on }[0, R_0] \text{ and }[R_0, R_1]
\end{equation}
Now $\bar{\rho}_p(R_0)=\rho_p(R_0)$, for if not, then we have $\bar{\rho}_p(r)-\rho_p(r)=\bar{\rho}_p(R_0)-\rho_p(R_0)$ for $r<R_0$. But neither function has a jump in value at $R_0$, so we apply the maximum principle again in $B_{R_1}\backslash B_{R_0}$ to conclude that $\bar{\rho}_p(r)-\rho_p(r)=\bar{\rho}_p(R_0)-\rho_p(R_0)$ on this set also. But we know that both functions vanish at $R_1$, so we have that $\bar{\rho}_p=\rho_p$ on $[0, R_1]$. Then $\rho_p$ cannot have a jump in its second derivative as $\bar{\rho}_p$ does not.

Using the extension of lemma 4.2 proved above, we can define $\bar{\rho}_e$ and  $\bar{\rho}_p$ in the same way for the case in which one of $\rho_e$ or $\rho_p$ have noncompact support taking $\lambda_f=0$ for the density with unbounded support, and still conclude that $\bar{\rho}_e$ and  $\bar{\rho}_p$ are $C^2$ on their supports. In the second part of the argument, assuming $\rho_p$ is the density with unbounded support, we have 
\begin{equation}\label{maximumprinciple}
\Delta(\bar{\rho}_{p}-\rho_{p})=0 \text{ on }[0, R_0] \text{ and }[R_0, \infty)
\end{equation}
We can still apply the maximum principle on $[0, R_0]$, and we can apply it to any $[R_0, R]$ $R<\infty$. So $\bar{\rho}_p-\rho_p$ is a constant on $[R_0, \infty)$. If this constant is not zero then $\bar{\rho}_p$ does not go to zero as $x\rightarrow\infty$, which clearly contradicts its definition.
\end{proof}
So the solutions to the system we found are $C^2$ on their support, and are thus solutions to the Euler-Lagrange equations (\ref{el1}) and (\ref{el2}). 

\section{The Structure of Minimizers}\label{synthesissection}

We now have quite a few results concerning the minimizers of (\ref{functional}), and here we combine them to see a more complete picture. First, we have theorems \ref{minexist} and \ref{compactminimizer} which tell us that as long as $N_e/N_p$ satisfies (\ref{5/3bounds}), (\ref{functional}) has a spherically symmetric minimizer such that the densities integrate to $N_e$ and $N_p$, and if the inequalities in (\ref{5/3bounds}) are sharp, the minimizer has compact support. Any such minimizer must satisfy the Euler-Lagrange equations, and we have seen in section \ref{ELsection} that given two central densities $(\alpha, \beta)$, these equations have a unique solution for which both densities are decreasing. These results fit together quite well, since we know that for a fixed $\alpha$, there is a single $\beta$ such that $u_p$ has unbounded support but is integrable, and a single $\beta$ such that $u_e$ has unbounded support but is integrable. Such solutions apparently correspond to the equalities of (\ref{5/3bounds}). If we call these two values $\beta_\alpha^h$ and $\beta_\alpha^l$, respectively, then every $\beta\in [\beta_\alpha^l, \beta_\alpha^h]$ will produce curves with compact support, and correspond to strict inequalities in (\ref{5/3bounds}).

We want to show that the minimizer is unique, and to do this, we must bridge the gap between the two approaches. It is clear that if we were given a minimizer for a pair $(N_e, N_p)$, we could use the central densities in solving the Euler-Lagrange equations, and therefore we recover the minimizer. So to show that the minimizer is unique, we need to show there is a bijection from $(N_e, N_p)$ to $(\alpha, \beta)$.
\begin{lemma}\label{bijection}
There is a bijective map from admissible values of $(N_e, N_p)$ to $(\alpha, \beta)$.
\end{lemma}
\begin{proof}
It is clear that given an $(\alpha,\beta)$, we can compute $N_e=\int u_e^{3/2}$ and $N_p=\int u_p^{3/2}$ by solving the ODE system. As a given pair of central densities has a unique solution, we just need to show that this map from $(\alpha,\beta)$ to $(N_e, N_p)$ is injective. To this end, we note that solutions to the bulk and atmospheric problem exhibit a scaling structure. That is, given a solution $(u_e(r), u_p(r))$, we can find another solution $(\theta(s), \eta(s))$ such that $\lambda u_e(r)=\theta(s)$ and $\lambda u_p(r)=\eta(s)$ where $r=as$ and $\lambda^{1/2}=a^2$. 

Then suppose we have $(\alpha_1, \beta_1)$ and $(\alpha_2, \beta_2)$ such that $\alpha_1\neq\alpha_2$ and $\beta_1\neq\beta_2$; if just one of these is an equality, we can use the type of argument given in lemma \ref{uniquespecial} to show that at least one of $\int u^{3/2}_{e,1}\neq \int u^{3/2}_{e,2}$ or $\int u^{3/2}_{p,1}\neq \int u^{3/2}_{p,2}$ must be true. Assume that $\alpha_1=\lambda\alpha_2$ for $\lambda>1$, so that we have a solution $(\theta, \eta)$ with $\theta(0)=\alpha_1$ and the relation $(\lambda u_{e,2}(r), \lambda u_{p,2})=(\theta(s), \eta(s))$. Then we have 
\begin{equation}
    \int u^{3/2}_{e,2}=4\pi\int_0^\infty r^2u^{3/2}_{e,2}(r)dr=4\pi\int_0^\infty r^2u^{3/2}_{e,2}(r)dr=\frac{4\pi a^3}{\lambda^{3/2}}\int_0^\infty s^2\theta^{3/2}(s)ds=a^{-3}\int \theta^{3/2}
\end{equation}
and similarly,
\begin{equation}
    \int u^{3/2}_{p,2}=a^{-3}\int \eta^{3/2}.
\end{equation}
A comparison principle type argument like what was used in lemma \ref{uniquespecial} tells us that either $\theta\leq u_{e,1}$ and $\eta \geq u_{p,1}$ or $\eta\leq u_{p,1}$ and $\theta\geq u_{e,1}$, depending on the relationship between $\lambda\beta_2$ and $\beta_1$. So we can conclude that we cannot have both $\int u^{3/2}_{e,1}= \int u^{3/2}_{e,2}$ and $\int u^{3/2}_{p,1}= \int u^{3/2}_{p,2}$. Therefore, the map from the central densities to the number of particles is injective.
\end{proof}

See \cite{HK} for further discussion and numerical solutions for these equations. 
\subsection{The Special Relativistic Case}
Let us rewrite (\ref{specialel3}) and (\ref{specialel4}) as 
\begin{equation}\label{specialEL3}
    \frac{1}{r^2}\frac{d}{dr}\left(r^2\frac{d}{dr}y_p\right)=\frac{1}{\alpha_pk_p^3}\left(1-G\frac{m_p^2}{q^2}\right)(y_p^2-1)^{3/2}-\frac{1}{\alpha_pk_e^3}\left(1+G\frac{m_pm_e}{q^2}\right)(y_q^2-1)^{3/2}
\end{equation}
and
\begin{equation}\label{specialEL4}
    \frac{1}{r^2}\frac{d}{dr}\left(r^2\frac{d}{dr}y_e\right)=-\frac{1}{\alpha_ek_p^3}\left(1+G\frac{m_pm_e}{q^2}\right)(y_p^2-1)^{3/2}+\frac{1}{\alpha_ek_e^3}\left(1-G\frac{m_q^2}{q^2}\right)(y_q^2-1)^{3/2},
\end{equation}
respectively.

Unlike the proof that minimizers of (\ref{functional}) exist, the proof of their uniqueness and the relation of the shape of the minimizers to $(N_e,N_p)$ does not so easily generalize to the special relativistic case. The main hindrance in directly generalizing the results of section \ref{ELsection} is that there are no special solutions which provide some basis for a type of comparison principle. 
\begin{lemma}
There are no decreasing solutions of the form $\rho_e=\hat{k}\rho_p$.
\end{lemma}
\begin{proof}
This is equivalent to saying there are no solutions of the form $(y_p^2-1)=k(y_q^2-1)$ for $k=\hat{k}\left(\frac{m_p}{m_e}\right)^3$. So we can conclude immediately that $k>0$, since we are only interested in positive solutions. Since we are further interested in decreasing solutions, we can plug $(y_p^2-1)=k(y_q^2-1)$ into the right hand side of (\ref{specialel3}) to get 
\begin{equation}
    \Delta y_p=\frac{4\pi}{m_pc^2}\left(\frac{\pi}{3}\right)\left(\frac{2c}{h}\right)^3\frac{1}{q^2}\left(m_p^3k\left(1-\frac{Gm^2_p}{q^2}\right)-m_e^3\left(1+\frac{Gm_pm_e}{q^2}\right)\right)(y_q^2-1)^{3/2}
\end{equation}
Since $y_p'(0)=0$, we can only have a decreasing solution if 
\begin{equation}
    k\leq \frac{m_e^3}{m_p^3}\frac{\left(1+\frac{Gm_pm_e}{q^2}\right)}{\left(1-\frac{Gm^2_p}{q^2}\right)}<1
\end{equation}
So we know that $0<k<1$.

Now we assume such a solution exists and derive a contradiction. Let $y_p=\sqrt{k(y_q^2-1)+1}$. Then (\ref{specialEL3}) becomes
\begin{equation}\label{newlaplacian}
    k\left[\frac{2}{r}\frac{y_e y_e'}{\sqrt{(y_q^2-1)k+1}}+\frac{(1-k)(y_e')^2+y_ey_e''(k(y_q^2-1)+1)}{((y_q^2-1)k+1)^{3/2}}\right]=(-E+Fk^{3/2})(y_q^2-1)^{3/2},
\end{equation}
for $E$ and $F$ appropriately defined. So just as in the proof of lemma (\ref{specialsolution}) we have a compatibility  condition on $k$ given by (\ref{newlaplacian}) and (\ref{specialEL4}), the latter of which we rewrite 
\begin{equation}\label{kcompatibility}
    y_e''+\frac{2}{r}y_e'=(-Ak^{3/2}+B)(y^2_e-1)^{3/2}
\end{equation}
This compatibility condition is not as straightforward though. In lemma (\ref{specialsolution}), our condition led us to an algebraic equation we could easily determine had a solution. Choosing the correct $k$ essentially eliminated one of the constraint equations, so then we only needed to solve a simple ODE. 

To proceed, we note that we can have neither $B=Ak^{3/2}$ nor $E=Fk^{3/2}$, for in the former case we would have $y_p$ be a constant and in the latter case we would have $y_e$ be a constant. Then combine (\ref{newlaplacian}) and (\ref{kcompatibility}) to get 
\begin{equation}
    \frac{2}{r}\frac{y_e y_e'}{\sqrt{(y_q^2-1)k+1}}+\frac{(1-k)(y_e')^2+y_ey_e''(k(y_q^2-1)+1)}{((y_q^2-1)k+1)^{3/2}}=\frac{-E+Fk^{3/2}}{k(-Ak^{3/2}+B)}\left[y_e''+\frac{2}{r}y_e'\right]
\end{equation}
This can be rearranged to get 
\begin{equation}\label{simplecompatibility}
    y_e''+\frac{2}{r}y_e'=\frac{(y_e')^2}{k(y_q^2-1)+1}\left[\frac{k(-Ak^{3/2}+B)(1-k)}{((y_q^2-1)k+1)^{1/2}(-E+k^{3/2}F)-y_ek(-Ak^{3/2}+B)}\right].
\end{equation}
Then we can recombine (\ref{simplecompatibility}) with (\ref{kcompatibility}) to get 
\begin{equation}\label{anothercompat}
    (y_e')^2=\frac{(y_q^2-1)^{3/2}}{k(1-k)}\left(k(y_q^2-1)+1\right)\left[(-E+Fk^{3/2})(k(y_q^2-1)+1)^{1/2}-y_ek(-Ak^{3/2}+B)\right].
\end{equation}
Then we have that 
\begin{equation}
    y_e'=\frac{(y_q^2-1)^{3/4}}{\sqrt{k(1-k)}}\left(k(y_q^2-1)+1\right)^{1/2}\left[(-E+Fk^{3/2})(k(y_q^2-1)+1)^{1/2}-y_ek(-Ak^{3/2}+B)\right]^{1/2}
\end{equation}
Differentiating (\ref{anothercompat}) we get 
\begin{align}
y_e''&=\frac{3(y_q^2-1)^{1/2}y_e}{2k(1-k)}(k(y_q^2-1)+1)\left[(-E+Fk^{3/2})(k(y_q^2-1)+1)^{1/2}-y_ek(-Ak^{3/2}+B)\right]\\
&+\frac{(y_q^2-1)^{3/2}y_e}{(1-k)}\left[(-E+Fk^{3/2})(k(y_q^2-1)+1)^{1/2}-y_ek(-Ak^{3/2}+B)\right] \nonumber\\
&+\frac{(y_q^2-1)^{3/2}}{2(1-k)}(k(y_q^2-1)+1)\left[(-E+Fk^{3/2})(k(y_q^2-1)+1)^{-1/2}y_e-(-Ak^{3/2}+B)\right]. \nonumber
\end{align}
 Then we can combine these last two equations with (\ref{kcompatibility}) to get 
\begin{align}
    -Ak^{3/2}+B&=\frac{3y_e}{2k(1-k)(y_q^2-1)}(k(y_q^2-1)+1)\left[(-E+Fk^{3/2})(k(y_q^2-1)+1)^{1/2}-y_ek(-Ak^{3/2}+B)\right]\\
&+\frac{y_e}{(1-k)}\left[(-E+Fk^{3/2})(k(y_q^2-1)+1)^{1/2}-y_ek(-Ak^{3/2}+B)\right] \nonumber\\
&+\frac{k(y_q^2-1)+1}{2(1-k)}\left[(-E+Fk^{3/2})(k(y_q^2-1)+1)^{-1/2}y_e-(-Ak^{3/2}+B)\right] \nonumber\\
&+ \frac{2}{r}\frac{(y_q^2-1)^{-3/4}}{\sqrt{k(1-k)}}\left(k(y_q^2-1)+1\right)^{1/2}\left[(-E+Fk^{3/2})(k(y_q^2-1)+1)^{1/2}-y_ek(-Ak^{3/2}+B)\right]^{1/2}\nonumber
\end{align}
the right side of which cannot be constant unless $y_e$ is constant. 
\end{proof}
Another major obstacle that prevents a direct generalization of lemma \ref{bijection} is that, as is easy to tell by direct computation, there is no scaling structure to these equations comparable to system (\ref{sphereEL1}) and (\ref{sphereEL2}). Despite these difficulties, numerics indicate that the structure of the solutions is exactly the same as in the Newtonian case, once one accounts for the limitations from (\ref{crudeestimate}), see \cite{HK}.

\section{Conclusion}\label{conclusion}
In this paper we have studied a two species Thomas-Fermi type model for a non-rotating non-neutral self-gravitating brown dwarf star in its ground state. This is an extension of previous results which assumed local neutrality and therefore were reduced to considering only a one fluid model.

In the Newtonian kinetic energy case, we gave a complete classification of the unique minimizers of (\ref{functional}) given that the total number of positively and negatively charged particles satisfied (\ref{5/3bounds}). We were also able to show the existence of minimizers in the special relativistic kinetic energy case.

Of course, these results still make a large number of unrealistic assumptions: at the very least, actual white dwarfs rotate and are not composed of only protons and electrons. In future work, we intend to extend this work to models removing these assumptions.  We believe it is however worthwhile to study the simplest nonneutral case before beginning to address more realistic, complicated models. 

\bibliographystyle{apsrev}

$\phantom{nix}$
\end{document}